\documentclass[pdflatex,sn-aps]{sn-jnl}

\usepackage{graphicx}%
\usepackage{multirow}%
\usepackage{amsmath,amssymb,amsfonts}%
\usepackage{amsthm,bbm}%
\usepackage{mathrsfs}%
\usepackage[title]{appendix}%
\usepackage{xcolor}%
\usepackage{textcomp}%
\usepackage{manyfoot}%
\usepackage{booktabs}%
\usepackage{algorithm}%
\usepackage{algorithmicx}%
\usepackage{algpseudocode}%
\usepackage{listings}%
\usepackage{mathtools}%
\usepackage{subcaption}%
\usepackage{tikz}%
\usetikzlibrary{arrows.meta}
\usepackage[linguistics]{forest}
\newcommand{\pc}{\mathit{PC}}
\tikzset{
  nd/.style={draw, rounded corners=2pt, inner sep=3pt,
    font=\small, minimum height=1.4em},
  hi/.style={nd, line width=1.5pt, fill=red!8},
  ok/.style={nd, fill=green!10},
  sub/.style={nd, densely dashed, fill=blue!5},
  unx/.style={nd, dotted},
  cl/.style={font=\scriptsize, fill=white, inner sep=1pt},
  arr/.style={-{Stealth[length=3pt]}},
  varr/.style={-{Stealth[length=3pt]}},
}
\usepackage{enumitem}
\usepackage{cleveref}
\usepackage{xspace}

\crefname{lemma}{lemma}{lemmas}
\Crefname{lemma}{Lemma}{Lemmas}

\theoremstyle{thmstyleone}%
\newtheorem{theorem}{Theorem}%
\newtheorem{proposition}[theorem]{Proposition}%
\newtheorem{lemma}[theorem]{Lemma}%

\theoremstyle{thmstyletwo}%
\newtheorem{example}{Example}%

\theoremstyle{thmstylethree}%
\newtheorem{definition}{Definition}%

\newcommand{\eldarica}{\textsc{Eld\-arica}\xspace}

\newcommand{\loat}{\textsc{LoAT}\xspace}
\newcommand{\klee}{\textsc{KLEE}\xspace}
\newcommand{\golem}{\textsc{Golem}\xspace}
\newcommand{\spacer}{\textsc{Z3/Spacer}\xspace}

\newcommand{\eldabs}{\texttt{CEGAR}\xspace}
\newcommand{\eldabsoff}{\texttt{CEGAR$^-$}\xspace}
\newcommand{\eldsymexbfs}{\texttt{FE}$_{\mathtt{BFS}}$\xspace}
\newcommand{\eldsymexdfs}{\texttt{FE}$_{\mathtt{DFS}}$\xspace}
\newcommand{\eldsymexsld}{\texttt{BE}$_{\mathtt{BFS}}$\xspace}

\newcommand{\false}{\ensuremath{\bot}}
\newcommand{\true}{\ensuremath{\top}}

\newcommand{\cmark}{\leftarrow} %
\newcommand{\csep}{\mathbin{\land}} %
\newcommand{\chc}[3]{#1 \cmark #2 \csep #3}
\newcommand{\chcfact}[2]{#1 \cmark #2}
\newcommand{\chcgoal}[2]{\false \cmark #1 \csep #2}

\newcommand{\achc}[3]{#1 &\cmark #2 \csep #3}
\newcommand{\achcfact}[2]{#1 &\cmark #2}
\newcommand{\achcgoal}[2]{\false &\cmark #1 \csep #2}

\newcommand{\calC}{\mathcal{C}}
\newcommand{\calF}{\mathcal{F}}
\newcommand{\calG}{\mathcal{G}}
\newcommand{\calM}{\mathcal{M}}
\newcommand{\calP}{\mathcal{P}}
\newcommand{\calR}{\mathcal{R}}
\newcommand{\calS}{\mathcal{S}}
\newcommand{\calT}{\mathcal{T}}
\newcommand{\calU}{\mathcal{U}}
\newcommand{\calX}{\mathcal{X}}

\newcommand{\HR}{\mathit{HR}}
\newcommand{\Conf}{\mathit{Conf}}

\newcommand{\resolve}[2]{%
  \dfrac{\displaystyle #1}{\displaystyle #2}}
\newcommand{\resolvecond}[3]{%
  \dfrac{\displaystyle #1}{\displaystyle #2}
  \;\; #3}

\lstdefinestyle{progstyle}{
  language=C,
  basicstyle=\small\ttfamily,
  keywordstyle=\bfseries\color{violet},
  commentstyle=\itshape\color{gray},
  stringstyle=\color{teal}, 
  numbers=none,
  frame=single,
  xleftmargin=1em,
  xrightmargin=1em,
  aboveskip=0.5em,
  belowskip=0.5em,
  morekeywords={assert,assume},
  emph={int,void},
  emphstyle=\color{blue!70!black}\bfseries,
  columns=flexible,
  keepspaces=true,
  literate={*}{{\textasteriskcentered}}1
}

\hypersetup{
  pdftitle={Symbolic Execution of Constrained Horn Clauses},
  pdfauthor={Johannes Weiser, Zafer Esen, Philipp R\"ummer},
  pdfsubject={Constrained Horn clauses and symbolic execution},
  pdfkeywords={Constrained Horn clauses, Symbolic execution, Software verification, Resolution, SLD, k-induction, Incorrectness logic}
}

\begin{document}

\title[Symbolic Execution of Constrained Horn Clauses]{Symbolic Execution of Constrained Horn Clauses}

\author[1]{\fnm{Johannes} \sur{Weiser}}
\email{johannes.weiser@tuwien.ac.at}

\author[2]{\fnm{Zafer} \sur{Esen}}
\email{zafer.esen@it.uu.se}

\author[2,3]{\fnm{Philipp} \sur{R\"ummer}}
\email{philipp.ruemmer@it.uu.se}

\affil[1]{\orgname{TU Vienna}, \orgaddress{\country{Austria}}}

\affil[2]{\orgname{Uppsala University}, \orgaddress{\country{Sweden}}}

\affil[3]{\orgname{University of Regensburg}, \orgaddress{\country{Germany}}}

\abstract{
Constrained Horn clauses (CHCs) are a fragment of first-order logic
widely used as an intermediate language for verification. We study \emph{constrained resolution} as a calculus for reasoning about the satisfiability of sets of CHCs. Constrained resolution forms a simple alternative to model checking algorithms such as CEGAR and IC3 and is here shown to exhibit complementary features. We show that, when CHCs encode programs, forward symbolic execution corresponds to positive unit hyper-resolution, and backward symbolic execution to SLD resolution, which are special cases of constrained resolution. We formalize forward and backward reasoning for linear and non-linear CHCs and transition systems, prove refutational completeness under fair strategies, and develop subsumption criteria for pruning redundant derivations. We moreover show that backward resolution with subsumption generalizes $k$-induction to non-linear CHCs, and that forward constrained resolution is related to the problem of constructing proofs in incorrectness logic. Lastly, we provide an experimental evaluation in the \eldarica CHC solver on benchmarks from CHC-COMP.
}

\keywords{Constrained Horn clauses, Symbolic execution,
          Software verification, Resolution, SLD, k-induction, Incorrectness logic}

\maketitle

\section{Introduction}\label{sec:intro}

Constrained Horn clauses (CHCs) are a fragment of
first-order logic that has found wide use as an
intermediate language for verification, with applications
ranging from security analysis for smart contracts to the verification of programs in various
languages~\citep{DBLP:conf/ndss/KalraGDS18,DBLP:conf/cav/BritikovZFAS24,DBLP:conf/pldi/GrebenshchikovLPR12,DBLP:conf/cav/GurfinkelKKN15,DBLP:conf/esop/0002T020,DBLP:journals/pacmpl/OgawaSU25,DBLP:conf/lpar/KahsaiKRS17,DBLP:conf/cav/WolffGEHRW25,DBLP:conf/pepm/SatoI019}. CHCs enable the modular construction of verification tools by decoupling the development of verification \emph{front-ends,} which take care of parsing and processing programs, applying the rules of some program logic, from CHC solvers as verification \emph{back-ends} that carry out the required semantic program analysis.
When encoding a program using CHCs, predicates typically represent the required annotations, such as state invariants, loop invariants, or function contracts, while CHCs encode program transitions as well as all further constraints on the annotations. A program is safe when the corresponding
set of CHCs is satisfiable.

To analyze CHCs, CHC solvers such as \eldarica~\citep{DBLP:conf/fmcad/HojjatR18}, \golem~\citep{DBLP:conf/cav/BlichaBS23}, \spacer~\citep{DBLP:conf/cav/KomuravelliGC14}, and \loat~\cite{DBLP:conf/cade/FrohnG22} implement a variety of model checking algorithms, for instance IC3~\citep{DBLP:conf/vmcai/Bradley11}, CEGAR~\citep{DBLP:conf/cav/ClarkeGJLV00}, loop acceleration~\citep{DBLP:conf/sas/FrohnG23}, bounded model checking~\citep{DBLP:conf/tacas/BiereCCZ99}, or transition power abstraction~\citep{DBLP:conf/tacas/BlichaFHS22}. As a simple alternative to those algorithms, in this article we study \emph{constrained resolution}~\citep{DBLP:journals/ai/Burckert94}, a logical proof rule that enables solvers to derive new clauses from a set of existing clauses in forward or backward direction. As our first insight, we show that strategies to apply constrained resolution correspond to \emph{symbolic execution,} a widely used approach in program analysis in which
some program variables are initialized with symbols that represent arbitrary values, rather than concrete ones~\citep{DBLP:conf/ibm/King74,DBLP:journals/cacm/King76}. Although more limited for proving the correctness of programs than algorithms used in software model checking, for instance IC3 or CEGAR, constrained resolution inherits the feature of symbolic execution of being very effective at discovering bugs in programs, which here corresponds to proving the \emph{unsatisfiability} of CHCs. Since constrained resolution does not rely on relatively expensive steps such as Craig interpolation~\cite{craig1957linear,DBLP:journals/corr/abs-2602-08532}, it can be implemented easily even in the context of complex theories like arrays or bit-vectors. We show that constrained resolution is also surprisingly effective at proving programs correct, reflecting experiences in software model checking~\cite{DBLP:conf/sigsoft/DietschHMNP17}.

As our second insight, we show that there is a correspondence between different resolution strategies and known forms of symbolic execution: \emph{forward} symbolic execution corresponds to \emph{constrained positive unit hyper-resolution,} whereas \emph{backward} symbolic execution~\cite{DBLP:conf/pldi/ChandraFS09,GulwaniJuvekar2009} corresponds to \emph{SLD resolution}~\cite{DBLP:journals/jacm/AptE82} known from logic programming~\cite{DBLP:journals/ai/KowalskiK71}.
Discussion of this correspondence goes back to the work of Kowalski, who observed that resolution steps over Horn clauses can be interpreted as procedure invocations~\cite{DBLP:conf/ifip/Kowalski74}, and to work in constraint logic programming~\citep{DBLP:conf/popl/JaffarL87}.
We extend this view to constrained resolution with background theories and make the connection to symbolic execution explicit.
We formalize both forward and backward resolution strategies over linear and non-linear CHCs, prove their correctness, and study subsumption for pruning redundant derivations.

Our third insight is that backward resolution combined with subsumption generalizes \emph{$k$-induction,} a proof rule for invariants in which the consecution condition assumes $k$ correct iterations of a loop~\cite{DBLP:conf/fmcad/SheeranSS00}. It is known that $k$-induction is a very effective rule for proving safety properties in both hardware and software verification (e.g., \cite{DBLP:conf/fmcad/SheeranSS00,DBLP:conf/sas/DonaldsonHKR11,DBLP:conf/cav/0001DW15}). We show that backward reasoning through constrained resolution has the effect of implicitly constructing the conditions that need to be checked when applying $k$-induction. A similar observation was made for the relationship between $k$-induction and backward symbolic execution in~\cite{DBLP:conf/sas/ChalupaS21}. Constrained resolution is more general than $k$-induction, however, since it applies also in the case of non-linear clauses.

As our fourth insight, we study the relationship between constrained resolution and \emph{incorrectness logic,} which is a logic that is dual to Hoare logic for reasoning about failing executions of programs~\cite{DBLP:journals/pacmpl/OHearn20}. We show that forward resolution proofs for CHCs encoding a program naturally translate to proofs in incorrectness logic, making it possible, in particular, to automatically construct \emph{non-ground} incorrectness proofs automatically. The relationship between CHC proofs and incorrectness reasoning also makes it possible to translate bespoke incorrectness proof rules, for instance for loop acceleration, to the CHC world.

The main contributions of this article are:
\begin{enumerate}[label=(\arabic*)]
    \item A formalization of symbolic execution as constrained resolution over CHCs, including subsumption criteria for CHCs (\Cref{sec:chc-symex} and \Cref{sec:subsumption}).
    \item A study of forward and backward exploration strategies over both linear and
    non-linear CHCs (\Cref{sec:forward} and \Cref{sec:backward}).
    \item Discussion how the $k$-induction proof rule can be reduced to backward resolution with clause subsumption (\Cref{sec:k-induction}) and of the connection between constrained resolution and incorrectness logic (\Cref{sec:incorrectness}).
    \item An experimental evaluation on CHC-COMP benchmarks (\Cref{sec:experiments}).
\end{enumerate}

\subsection{Motivating Example}
\label{sec:motivating}

Consider the program in \Cref{fig:program} (left side). The variable~\texttt{x} is initialized to a nondeterministic value, and the \texttt{while (*)} loop, which may iterate any number of times, alternates between adding~2 and subtracting~1 in each iteration. The program refers to a  placeholder predicate~\texttt{P(x,~x0)}, which will later be replaced with the
formulas \texttt{x = x0} and \texttt{x >= x0}. We can observe that the program is \emph{safe} if \texttt{P(x,~x0)} is interpreted as \texttt{x >= x0}: the loop body increases~\texttt{x} by a net~$1$ per two iterations, so \texttt{x >= x0} holds whenever the loop is left. If \texttt{P(x,~x0)} is interpreted as \texttt{x = x0}, the assertion fails in every execution in which the loop iterates at least once.

The translation of the program to a set of CHCs is shown in \Cref{fig:program} (right side). The encoding uses four clauses: a \emph{fact}~\eqref{cl:init}, representing the set of states at program entry, \emph{rules}~\eqref{cl:add} and~\eqref{cl:sub} representing the loop body, and a \emph{goal}~\eqref{cl:goal} encoding the assertion. The predicate~$q$ represents a loop invariant and has three arguments: the current value of~$x$, the flag~$b$, and the initial value of~$x$ stored in~$x_0$.

To analyze the CHCs, we use  \emph{constrained resolution}~\citep{DBLP:journals/ai/Burckert94}, a variant of resolution~\citep{DBLP:journals/jacm/Robinson65} in which syntactic unification is replaced by equality constraints:
\begin{equation*}
\resolve{
  \chc{p(\bar{s})}{B_1}{\varphi_1}
  \qquad
  \chc{H}{p(\bar{t}) \csep B_2}{\varphi_2}
}{
  \chc{H}{B_1 \csep B_2}{\bar{s} = \bar{t} \csep \varphi_1 \csep \varphi_2}
}
\end{equation*}
The rule derives a new clause by replacing the matched atom~$p(\bar{t})$ in the body of the right premise with the body and constraints of the left premise. The equality $\bar{s} = \bar{t}$ is used instead of syntactic unification because clause terms may involve theory-interpreted symbols, such as in $x + 2$, whose equality is delegated to a theory solver. The clause bodies~$B_1, B_2$ are (possibly empty) conjunctions of atoms,
and $\varphi_1, \varphi_2$ are constraints over the background theory; in the examples, the clauses are formulated over linear integer arithmetic (LIA) as the background theory. 

\begin{figure}[tb]
\centering
\begin{minipage}[t]{0.38\linewidth}
\begin{lstlisting}[style=progstyle,numbers=left,
  numberstyle=\tiny\color{gray},frame=none,
  xleftmargin=1.5em,xrightmargin=0pt]
int x := *, b := 0;
int x0 := x;
while (*) {
  if (b = 0) x := x + 2;
  else       x := x - 1;
  b := 1 - b;
}
assert(P(x, x0));
\end{lstlisting}
\end{minipage}
\hfill
\begin{minipage}[t]{0.58\linewidth}
\vspace{0pt}
\begin{align}
  q(x, b, x_0) &\cmark x = x_0 \csep b = 0
    \label{cl:init} \\[2pt]
  q(x + 2, 1, x_0) &\cmark q(x, b, x_0) \csep b = 0
    \label{cl:add} \\[2pt]
  q(x - 1, 0, x_0) &\cmark q(x, b, x_0) \csep b = 1
    \label{cl:sub} \\[2pt]
  P(x, x_0) &\cmark q(x, b, x_0)
    \label{cl:goal}
\end{align}
\end{minipage}
\caption{A program with a nondeterministic loop that either
  adds~$2$ to or subtracts~$1$ from~\texttt{x} in each iteration
  (left), with its CHC encoding (right). The construct
  \texttt{while~(*)} denotes a loop that may execute any number
  of times. The variable~\texttt{x} is initialized to a
  nondeterministic value, and~\texttt{x0} records its initial
  value. Predicate $q(x, b, x_0)$ represents the set of states
  at the loop head, with~$x_0$ tracking the initial value
  of~$x$. The assertion contains a placeholder formula
  \texttt{P(x,~x0)} in the program; in~\eqref{cl:goal} the
  formula is represented by the predicate $P(x, x_0)$, which is
  instantiated as $x = x_0$ in
  Figure~\ref{fig:forward-unsafe}, and as $x \geq x_0$ in
  Figures~\ref{fig:forward-safe},~\ref{fig:backward}.}
\label{fig:program}
\end{figure}

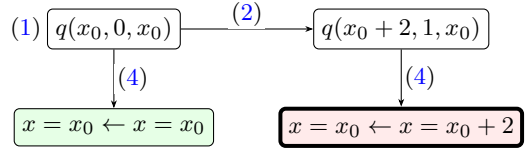
\begin{figure}[tb]
\begin{center}
\begin{tikzpicture}
\node[nd] (a1) at (0, 0) {$q(x_0, 0, x_0)$};
\node[nd] (a2) at (3.8, 0) {$q(x_0 + 2, 1, x_0)$};
\node[cl] at ([xshift=-0.8em]a1.west) {\eqref{cl:init}};
\draw[arr] (a1) -- (a2)
  node[cl, midway, above] {\eqref{cl:add}};
\node[ok] (b1) at (0, -1.3) {$x = x_0 \cmark x = x_0$};
\node[hi] (b2) at (3.8, -1.3) {$x = x_0 \cmark x = x_0 + 2$};
\draw[varr] (a1) -- (b1)
  node[cl, midway, right] {\eqref{cl:goal}};
\draw[varr] (a2) -- (b2)
  node[cl, midway, right] {\eqref{cl:goal}};
\end{tikzpicture}
\end{center}
\caption{Forward execution with $P(x, x_0) \equiv (x = x_0)$.
  The clause $x = x_0 \cmark x = x_0 + 2$ (thick red border) is unsat, corresponding to an assertion violation after one iteration.}
\label{fig:forward-unsafe}
\end{figure}

Choosing $P(x, x_0) \equiv (x = x_0)$, the program has a bug: the assertion can fail because after one loop iteration the value of~$x$ changes and is no longer equal to $x_0$. Such a failing execution can be found using \emph{forward symbolic execution,} running the program with symbolic initial value~$x$ and checking whether any feasible paths to a failing assertion can be found. The resolution diagram in \Cref{fig:forward-unsafe} illustrates how constrained resolution mimics forward symbolic execution, discovering the bug in two steps. Each arrow in the figure represents an application of the constrained resolution rule; for readability, the diagram shows simplified constraints in which equalities have been inlined. Resolving~\eqref{cl:init} with~\eqref{cl:add} yields the derived clause $q(x_0 + 2, 1, x_0)$; resolving this clause with~\eqref{cl:goal} produces $x = x_0 \cmark x = x_0 + 2$. Since $x_0 + 2 \neq x_0$ and the variables~$x, x_0$ are implicitly universally quantified, the clause obtained in this last step is unsatisfiable; the resolution proof shows that the CHCs~\eqref{cl:init}, \eqref{cl:add}, \eqref{cl:goal} are unsatisfiable, which is equivalent to the program having failing executions.

There is a direct correspondence between the applied resolution steps and the execution of program statements: the unit clause~$q(x_0, 0, x_0)$ is the symbolic initial state of the program, while $q(x_0 + 2, 1, x_0)$ is the symbolic state after one loop iteration. Branches in the symbolic execution, i.e., symbolic states in which multiple program statements are feasible, correspond to branches in the resolution diagram. The forward resolution diagram corresponds exactly to the \emph{symbolic execution tree} of a program~\cite{DBLP:journals/cacm/King76}.

\begin{figure}[tb]
\begin{center}
\begin{tikzpicture}[every node/.append style={font=\footnotesize}]
\node[nd] (a1) at (0, 0) {$q(x_0, 0, x_0)$};
\node[nd] (a2) at (3.0, 0) {$q(x_0{+}2, 1, x_0)$};
\node[nd] (a3) at (6.0, 0) {$q(x_0{+}1, 0, x_0)$};
\node[nd] (a4) at (9.0, 0) {$q(x_0{+}3, 1, x_0)$};
\node[font=\normalsize, gray] at (10.9, 0) {$\cdots$};
\node[cl] at ([xshift=-0.8em]a1.west) {\eqref{cl:init}};
\draw[arr] (a1) -- (a2)
  node[cl, midway, above] {\eqref{cl:add}};
\draw[arr] (a2) -- (a3)
  node[cl, midway, above] {\eqref{cl:sub}};
\draw[arr] (a3) -- (a4)
  node[cl, midway, above] {\eqref{cl:add}};
\draw[arr] (a4) -- (10.5, 0);
\node[ok] (b1) at (0, -1.3) {$x {\geq}\, x_0 \cmark x {=}\, x_0$};
\node[ok] (b2) at (3.0, -1.3) {$x {\geq}\, x_0 \cmark x {=}\, x_0{+}2$};
\node[ok] (b3) at (6.0, -1.3) {$x {\geq}\, x_0 \cmark x {=}\, x_0{+}1$};
\node[ok] (b4) at (9.0, -1.3) {$x {\geq}\, x_0 \cmark x {=}\, x_0{+}3$};
\node[font=\normalsize, gray] at (10.9, -1.3) {$\cdots$};
\draw[varr] (a1) -- (b1)
  node[cl, midway, right] {\eqref{cl:goal}};
\draw[varr] (a2) -- (b2)
  node[cl, midway, right] {\eqref{cl:goal}};
\draw[varr] (a3) -- (b3)
  node[cl, midway, right] {\eqref{cl:goal}};
\draw[varr] (a4) -- (b4)
  node[cl, midway, right] {\eqref{cl:goal}};
\end{tikzpicture}
\end{center}
\caption{Forward execution with $P(x, x_0) \equiv (x \geq x_0)$.
  Every check holds (green), but the chain grows indefinitely.}
\label{fig:forward-safe}
\end{figure}
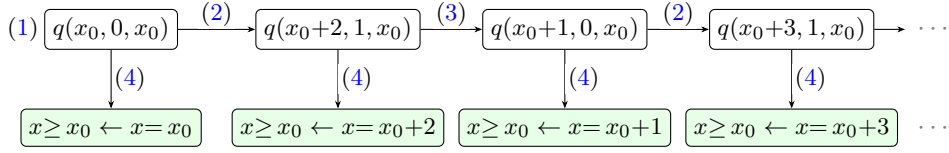

We now consider the safe case $P(x, x_0) \equiv (x \geq x_0)$, for which a resolution diagram is shown in \Cref{fig:forward-safe}. We can observe that none of the derived clauses is unsatisfiable, but the chain of derived clauses grows indefinitely. Forward resolution can conclude neither safety nor unsafety for this program.

\begin{figure}[tb]
\begin{center}
\begin{tikzpicture}[every node/.append style={font=\scriptsize}]
\node[nd] (a1) at (0, 0)
  {$x \geq x_0 \cmark q(x, b, x_0)$};
\node[cl] at ([xshift=-0.8em]a1.west) {\eqref{cl:goal}};
\node[nd] (a2) at (4.3, 0)
  {$x \geq x_0{+}1 \cmark q(x, b, x_0) \csep b {=} 1$};
\node[sub] (a3) at (9.2, 0)
  {$x \geq x_0{-}1 \cmark q(x, 0, x_0)$};

\draw[arr] (a1) -- (a2)
  node[cl, midway, above] {\eqref{cl:sub}};
\draw[arr] (a2) -- (a3)
  node[cl, midway, above] {\eqref{cl:add}};

\node[sub] (a0) at (4.3, 1.3)
  {$x \geq x_0{-}2 \cmark q(x, b, x_0) \csep b {=} 0$};
\draw[varr] (a1) -- (a0)
  node[cl, midway, right] {\eqref{cl:add}};

\node[ok] (b1) at (0, -1.3)
  {$x \geq x_0 \cmark x {=}\, x_0 \csep b {=} 0$};
\node[ok] (b2) at (4.3, -1.3)
  {$x \geq x_0{+}1 \cmark x {=}\, x_0 \csep b {=} 0 \csep b {=} 1$};

\draw[varr] (a1) -- (b1)
  node[cl, midway, right] {\eqref{cl:init}};
\draw[varr] (a2) -- (b2)
  node[cl, midway, right] {\eqref{cl:init}};
\end{tikzpicture}
\end{center}
\caption{Backward execution with $P(x, x_0) \equiv (x \geq x_0)$.
  Clauses in dashed nodes are subsumed by~\eqref{cl:goal}. All derivable clauses are satisfiable, proving that the set of clauses is satisfiable and that the program is safe.}
\label{fig:backward}
\end{figure}
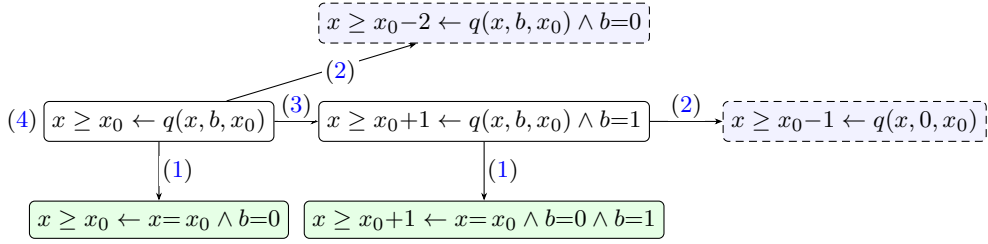

\Cref{fig:backward} applies \emph{backward resolution} to the same property $P(x, x_0) \equiv (x \geq x_0)$. Starting from the goal~\eqref{cl:goal}, resolving with~\eqref{cl:sub} yields $x \geq x_0 + 1 \cmark q(x, b, x_0) \csep b = 1$: for any predecessor state where $b = 1$, i.e., where the subtraction branch was taken, $x$ must be at least $x_0 + 1$ to ensure $x - 1 \geq x_0$ after the step. Resolving the original clause with~\eqref{cl:add} yields $x \geq x_0 - 2 \cmark q(x, b, x_0) \csep b = 0$ (dashed node). This clause is \emph{subsumed} by~\eqref{cl:goal}: the body of~\eqref{cl:goal} is more general, since it places no constraint on~$b$, and its constraint $x \geq x_0$ entails $x \geq x_0 - 2$. Every state that can reach an assertion violation through the subsumed clause can also reach it through~\eqref{cl:goal}, so the subsumed clause adds no information. The same holds for $x \geq x_0 - 1 \cmark q(x, 0, x_0)$, obtained by resolving the second derived clause with~\eqref{cl:add}. We define subsumption formally in \Cref{sec:subsumption}.

The two checks against~\eqref{cl:init} (green nodes) confirm that no contradiction can be derived when taking the initial state into account. Resolving~\eqref{cl:goal} with~\eqref{cl:init} gives the clause $x \geq x_0 \cmark x = x_0 \csep b = 0$, which holds since $x_0 \geq x_0$. Resolving the second derived clause with~\eqref{cl:init} gives $x \geq x_0 + 1 \cmark x = x_0 \csep b = 0 \csep b = 1$, which is vacuously valid since $b = 0 \land b = 1$ is unsatisfiable. All derived clauses are either subsumed or discharged, so backward resolution terminates and proves safety.

Each backward resolution step unwinds the transition relation by one step from the assertion. This corresponds to $k$-induction~\citep{DBLP:conf/fmcad/SheeranSS00,DBLP:conf/sas/ChalupaS21}: the base case verifies that the property holds at the initial state, and the induction step verifies that no new proof obligations arise beyond those already covered. In this example, two backward steps suffice ($k = 2$). We formalize this correspondence in \Cref{sec:k-induction}.

\section{Preliminaries}\label{sec:prelim}
\subsection{First-Order Logic}\label{sec:fol}
We work in a many-sorted first-order logic setting.
A signature is a tuple 
$\Sigma = (\calS, \calF, \calP, \delta)$, where
$\calS$ is a set of sorts, $\calF$ is a set of function
symbols, $\calP$ is a set of predicate symbols, and
$\delta$ provides the sort signatures for functions and predicates.
We assume a countably infinite set of sorted variables $\calX$.
A \emph{term} is either a variable or the application of a function symbol
to terms of appropriate sorts. 
An atomic formula (or \emph{atom}) is the application of a 
predicate symbol to terms of appropriate sorts. 
An atom is \emph{ground} if it contains no variables.
Formulas are constructed from atoms using
standard logical connectives ($\neg, \land, \lor, \rightarrow$) and
quantifiers ($\forall, \exists$).
A \emph{literal} is an atom or its negation.  
A \emph{clause} is a disjunction of $n \geq 0$ literals. 
If $n = 0$, the clause is the \emph{empty clause} (\false).
If $n = 1$, the clause is a \emph{unit clause}.
A unit clause with a non-negated literal is \emph{positive},
and one with a negated literal is \emph{negative}.
A clause that contains only negated literals is a goal clause.
A sentence is a formula with no free-variables, i.e.\ every variable is quantified. 

A $\Sigma$-structure $M = (\calU, I)$ is given by a domain $\calU = \bigcup_{S \in \calS} U_S$ and an interpretation $I$ s.t.\ for any sort $S \in \calS$, $I(S) = U_S$, for any predicate symbol $P \in \calP$ of type $\prod_{i=1}^n S_i$, $I(P) \subseteq \prod_{i=1}^n U_{S_i}$ and for any function symbol $f \in \calF$ of type $\prod_{i=1}^n S_i \to S$, $I(f)$ is in $\prod_{i=1}^n U_{S_i} \to U_S$.
If $\Sigma_1 \subseteq \Sigma_2$, we can extend each $\Sigma_1$-structure $M$ to a $\Sigma_2$ structure $M'$ in the natural way by adding the necessary interpretations of sorts, function symbols and predicate symbols. 

A variable assignment $\beta$ is a function from $\calX$ to $\calU$ mapping variables to elements of the appropriate sorts. 
Given a structure $M$ and a variable assignment $\beta$, the evaluation function $val_{M, \beta}$, mapping terms to elements of the domain and formulas to $\{true, false\}$ is defined as usual.
If $F$ is a formula and $val_{M, \beta}(F) = true$, we write $M, \beta \vDash F$.
If $M,\beta \vDash F$ for all variable assignments $\beta$, then $M \vDash F$. 
Note that the evaluation of a sentence $\gamma$ is independent of the variable assignment $\beta$ as there are no free variables.
If $\Gamma$ is a set of sentences and $M \vDash \gamma$ for every $\gamma \in \Gamma$, we write $M \vDash \Gamma$.

A \emph{theory}~$\calT$ is given by a signature~$\Sigma_\calT$ and a class~$\mathcal{M}_\calT$ of $\Sigma_\calT$-structures. 
Given a signature $\Sigma^* \supseteq \Sigma_\calT$, a set of $\Sigma^*$-sentences~$\Gamma$ and a $\Sigma^*$-formula~$F$, we write $\Gamma \vDash_\calT F$ if for every $M \in \calM_\calT$ and every extension $M'$ of $M$ to a $\Sigma^*$-structure, it holds that $M' \vDash F$. 
Throughout the article, we fix a background theory~$\calT$ that includes the equality operator~$=$ and write $\vDash$ as shorthand for $\vDash_\calT$.

\subsection{Constrained Horn Clauses (CHCs)}\label{sec:chc}
A \emph{Constrained Horn Clause} (CHC) is a clause with at most
one positive literal and a constraint over the background theory $\calT$. In implication form, a CHC
is written as:
\begin{equation*}
    \forall \bar{x}.\ \chc{H}{B}{\varphi}
\end{equation*}
where $H$ (the \emph{head}) is either an atom or $\false$,
$B$ (the \emph{body}) is a possibly empty conjunction of
uninterpreted atoms 
$p_1(\bar{t}_1) \land \cdots \land p_n(\bar{t}_n)$,
$\varphi$ is a $\calT$-constraint, and $\bar{x}$ is the tuple
of all variables occurring in the clause.

All variables in a CHC are implicitly universally quantified,
we usually omit the~$\forall\bar{x}$ prefix for brevity.
Whenever two CHCs are involved in a resolution
step, we assume that their variable sets are disjoint.  This convention is
implicit in all inference rules presented in this article.

CHCs are classified as follows. A \emph{unit clause} is a clause~$\chc{H}{B}{\varphi}$ in which the head~$H$ and body~$B$ together contain exactly one atom. A \emph{fact} or \emph{positive
unit clause} is a unit clause whose body contains no atoms, i.e.,
it has the form $\chcfact{H}{\varphi}$. A \emph{rule} is a CHC
whose head is an atom (not $\false$) and whose body contains
at least one atom. A \emph{goal} is a clause whose head is~$\false$, and a \emph{negative unit clause} is a goal that is also a unit clause.
We sometimes write goals as
$\psi \cmark B$, where $\psi$ is a constraint, which is equivalent to
$\chcgoal{B}{\neg\psi}$.

A CHC is \emph{linear} if its body contains at most one
uninterpreted atom, and \emph{non-linear} otherwise.
A set of CHCs is linear if each of its elements is linear.
Non-linear CHCs arise naturally when encoding programs with
function calls: the body of the clause at the call site contains
both the predicate for the called function and the predicate for
the continuation after the call. Non-recursive functions can be inlined to recover linearity, but recursive functions inherently result in non-linear clauses.

We work with a \emph{strict} notion of semantics~\cite{BjornerGMR15} and say that a set~$\mathcal{C}$ of CHCs is \emph{$\calT$-satisfiable} (or just \emph{satisfiable}) if $\mathcal{C} \not\models \mathit{false}$, i.e., if it is satisfiable in the sense of first-order logic.

\subsection{Calculi and Resolution}\label{sec:resolution}

The article studies several variants of the classical resolution calculus~\citep{DBLP:journals/jacm/DavisP60,DBLP:journals/jacm/Robinson65} for CHCs. In general, we say that a calculus is \emph{sound} if every derived clause is a
logical consequence of the premises. Deriving the empty clause
($\false$) then establishes that the original set of clauses is
unsatisfiable. The calculus is \emph{refutation-complete}
(referred to as \emph{complete} in the rest of this article) if,
for every unsatisfiable set of clauses, the empty clause can be
derived by finite applications of the calculus rules.

Resolution is an inference system for establishing the
unsatisfiability of a set of clauses. The propositional
resolution principle is based on a simple logical observation:
given two clauses $C \lor A$ and $D \lor \neg A$, where $A$ is an
atomic proposition, one can infer the \emph{resolvent} $C \lor D$
since $A$ and $\neg A$ cannot be true at the same time. If repeatedly
resolving clauses eventually produces the empty clause~$\false$,
the original set of clauses is proven unsatisfiable.

\paragraph{First-Order Resolution}
In first-order logic, the two literals being resolved need not be
identical; it suffices that they can be made identical through \emph{unification}~\citep{DBLP:journals/jacm/Robinson65}, 
a syntactic operation that substitutes variables with terms. 
Let $A\sigma$ denote the application of a substitution $\sigma$ to an atom $A$. A substitution $\sigma$ is a \emph{unifier} of two atoms $A$ and $A'$ if $A\sigma = A'\sigma$. It is a \emph{most general unifier} (MGU), denoted $\sigma = \mathit{mgu}(A, A')$, if for every other unifier $\theta$ of $A$ and $A'$, there exists a substitution $\gamma$ such that $\theta = \sigma\gamma$.
The \emph{first-order resolution} rule
generalizes propositional resolution using unification:
\begin{equation}\label{eq:fores}
    \frac{C \lor A \quad D \lor \neg A'}{(C \lor D)\sigma}
    \qquad \sigma = \mathit{mgu}(A, A')
\end{equation}
where the two premise clauses are assumed to share no variables.

\begin{example}
Consider the clauses $r(x) \lor p(f(x), y)$ and
$s(y) \lor \neg p(f(a), b)$. The atoms $p(f(x), y)$ and
$p(f(a), b)$ have the most general unifier
$\sigma = \{x \mapsto a,\, y \mapsto b\}$. Applying~\eqref{eq:fores} to these clauses using $\sigma$ yields the resolvent $r(a) \lor s(b)$.
\end{example}

\paragraph{Constrained Resolution}

When clauses involve a background theory~$\calT$, terms
may contain theory-interpreted function symbols 
(e.g., terms like $x+1$, $select(a, i)$) for which syntactic unification
is not sufficient.
\emph{Constrained resolution}~\citep{DBLP:journals/ai/Burckert94} extends first-order
resolution by replacing the MGU computation with an equality constraint:
instead of computing a unifier, the rule adds a constraint that the
arguments must be equal, and delegates the task of checking the satisfiability of this equality to a $\calT$-solver, usually an SMT solver. In the context of CHCs, the rule becomes:
\begin{equation}\label{eq:chcres}
\resolve{
  \chc{p(\bar{s})}{B_1}{\varphi_1}
  \qquad
  \chc{H}{p(\bar{t}) \csep B_2}{\varphi_2}
}{
  \chc{H}{B_1 \csep B_2}{\bar{s} = \bar{t} \csep \varphi_1 \csep \varphi_2}
}
\end{equation}
where $\varphi_1$ and $\varphi_2$ are $\calT$ constraints,
$\bar{s} = \bar{t}$ is the equality constraint, and the clauses must again not share any variables.

\begin{example}
Consider the fact $q(x, b, x_0) \cmark x {=} x_0 \csep b {=} 0$ (\eqref{cl:init} from \Cref{fig:program}) and the rule $q(x'{+}2, 1, x_0') \cmark q(x', b', x_0') \csep b' {=} 0$ (\eqref{cl:add} with renamed variables). Constrained resolution matches the head of~\eqref{cl:init} against the body atom of~\eqref{cl:add}, adding the equality constraints $x' {=} x \csep b' {=} b \csep x_0' {=} x_0$. The resolvent is:
\[
q(x'{+}2, 1, x_0') \cmark
 x {=} x_0 \csep b {=} 0 \csep
 x' {=} x \csep b' {=} b \csep x_0' {=} x_0 \csep b' {=} 0
\]
Simplifying by substituting the equalities yields the fact $q(x_0{+}2, 1, x_0) \cmark \true$, which is the first derived clause in \Cref{fig:forward-unsafe}.
\end{example}

\section{Symbolic Execution as Resolution}\label{sec:chc-symex}
Symbolic execution explores a program by maintaining a
\emph{symbolic store,} which is a mapping from program variables to
symbolic expressions, and a \emph{path condition} that
accumulates branch conditions along the explored
path~\citep{DBLP:journals/cacm/King76}. At each
branch point, execution splits into one child per
branch, producing a tree in which every path from the root to a leaf
represents a class of concrete executions sharing the same
control flow. At assertion points, it is checked whether the asserted
property can be violated in the symbolic program state; a positive result yields a failing program execution.

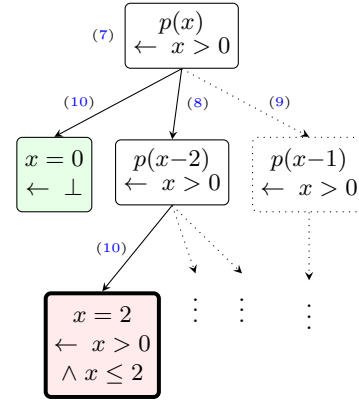
\begin{figure}[tb]
\begin{center}
\begin{minipage}[t]{0.47\linewidth}
\begin{lstlisting}[style=progstyle,numbers=left,
  numberstyle=\tiny\color{gray},frame=none,
  xleftmargin=1.5em,xrightmargin=0pt]
int x := *;
assume(x > 0);
while (x > 0) {
  if (*) x := x - 2;
  else   x := x - 1;
}
assert(x = 0);
\end{lstlisting}
\end{minipage}%
\hfill
\textcolor{gray!40}{\vrule}%
\hfill
\begin{minipage}[t]{0.47\linewidth}
\vspace{0pt}
\begin{align}
  p(x) &\cmark x > 0
    \label{se:init} \\
  p(x - 2) &\cmark p(x) \csep x > 0
    \label{se:step1} \\
  p(x - 1) &\cmark p(x) \csep x > 0
    \label{se:step2} \\
  x = 0 &\cmark p(x) \csep x \leq 0
    \label{se:goal}
\end{align}
\end{minipage}

\medskip
\begin{minipage}[t]{0.47\linewidth}
\centering
\begin{forest}
  for tree={
    draw, rounded corners=2pt, inner sep=3pt,
    font=\scriptsize, align=center,
    edge={->, >=stealth},
    l sep=0.9cm, s sep=0.3cm,
  }
  [{$[x \mapsto x_0]$\\$\pc\colon x_0 > 0$}
    [{$[x \mapsto x_0]$\\$\pc\colon x_0 > 0$\\$\csep\; x_0 \leq 0$}, ok,
      edge label={node[midway,font=\tiny,fill=white,inner sep=1pt]
        {$x{=}0\;?$}}]
    [{$[x \mapsto x_0{-}2]$\\$\pc\colon x_0 > 0$}
      [{$[x \mapsto x_0{-}2]$\\$\pc\colon x_0 > 0$\\$\csep\; x_0 \leq 2$},
        hi,
        edge label={node[midway,font=\tiny,fill=white,inner sep=1pt]
          {$x{=}0\;?$}}]
      [{$\vdots$}, draw=none,
        edge={dotted,->,>=stealth}]
      [{$\vdots$}, draw=none,
        edge={dotted,->,>=stealth}]
    ]
    [{$[x \mapsto x_0{-}1]$\\$\pc\colon x_0 > 0$}, unx,
      edge={dotted,->,>=stealth},
      [{$\vdots$}, draw=none,
        edge={dotted,->,>=stealth}]
    ]
  ]
\end{forest}

\smallskip
{\footnotesize (a) Depth-first symbolic execution.}
\end{minipage}%
\hfill
\textcolor{gray!40}{\vrule}%
\hfill
\begin{minipage}[t]{0.47\linewidth}
\centering
\begin{forest}
  for tree={
    draw, rounded corners=2pt, inner sep=3pt,
    font=\scriptsize, align=center,
    edge={->, >=stealth},
    l sep=0.9cm, s sep=0.3cm,
  }
  [{$p(x)$\\$\cmark\; x > 0$}, name=root
    [{$x = 0$\\$\cmark\; \false$}, ok,
      edge label={node[midway,left,font=\tiny,xshift=-0.5em]{\eqref{se:goal}}}
    ]
    [{$p(x{-}2)$\\$\cmark\; x > 0$},
      edge label={node[midway,right,font=\tiny]{\eqref{se:step1}}}
      [{$x = 2$\\$\cmark\; x > 0$\\$\csep\; x \leq 2$}, hi,
        edge label={node[midway,left,font=\tiny]{\eqref{se:goal}}}
      ]
      [{$\vdots$}, draw=none,
        edge={dotted,->,>=stealth}]
      [{$\vdots$}, draw=none,
        edge={dotted,->,>=stealth}]
    ]
    [{$p(x{-}1)$\\$\cmark\; x > 0$}, unx,
      edge={dotted,->,>=stealth},
      edge label={node[midway,right,font=\tiny,xshift=0.5em]{\eqref{se:step2}}}
      [{$\vdots$}, draw=none,
        edge={dotted,->,>=stealth}]
    ]
  ]
  \node[font=\tiny] at ([xshift=-0.8em]root.west)                             {\eqref{se:init}}; 
\end{forest}

\smallskip
{\footnotesize (b) Depth-first forward unit resolution.}
\end{minipage}
\end{center}

\caption{(a)~Each node shows a symbolic store and path
  condition (PC). The PC accumulates branch conditions.
  At green and red nodes, it is checked whether the property $x = 0$ holds in the symbolic state. The leftmost green node has an unsatisfiable PC, corresponding to exiting the loop immediately, which means that the assertion cannot be reached on this path. The thick red node shows a feasible path on which the asserted property can be
  violated, a counterexample is $x_0 = 1, x = -1$.
  (b)~Clauses encode the symbolic store as arguments of predicates and the PC as the constraint.
  An infeasible path yields clauses with unsatisfiable constraint. The goal
  clause~\eqref{se:goal} is resolved with every fact,
  including \eqref{se:init} and any resolvents, and the search terminates
  if a contradiction is derived (thick red node). Both for symbolic execution and depth-first forward unit resolution, completeness depends on the strategy: for instance, DFS would
  not terminate if the rightmost branch (line 5 in the program, or \eqref{se:step2}) were chosen in every iteration.}
\label{fig:symex-resolution}
\end{figure}

Symbolic execution is illustrated in \Cref{fig:symex-resolution}~(a) for a program that is similar to the motivating example. The initial symbolic program state is represented by the root of the tree, whereas the intermediate nodes represent possible states reached by the program during its execution. Nodes with unsatisfiable path condition represent \emph{infeasible} paths that cannot occur in concrete program execution. Red nodes mark feasible symbolic states at which an assertion can be violated. In general, the symbolic execution tree of a program is finitely branching, but might contain infinite paths corresponding to non-terminating program executions.

\Cref{fig:symex-resolution}~(b) shows the resolution diagram corresponding to the symbolic execution tree. In the CHC encoding assumed in this example, each predicate represents the
set of states that must hold at a program location, with the
predicate arguments corresponding to the program variables in
scope at that location. Facts represent the program entry points and state preconditions or fixed initial values of variables. Rules encode transitions between program locations, while goals encode the asserted program properties. Comparing the symbolic execution tree with the resolution diagram, it can be seen that the symbolic program states correspond to CHCs that can be derived using \emph{constrained unit resolution,} i.e., the constrained resolution rule restricted to applications in which at least one of the premises is a unit clause. The initial and intermediate symbolic program states translate to \emph{positive} unit clauses, while clauses without any predicate atoms occur when checking the program assertion. An assertion failure corresponds to an \emph{unsatisfiable} clause being derived, i.e., a closed resolution proof showing that the set of CHCs is unsatisfiable. One \emph{difference} between symbolic execution and constrained resolution is that derived clauses do not capture the relationship between initial and intermediate program states: in symbolic execution, the value of variable~$x$ is expressed in terms of the initial value~$x_0$, whereas the derived clauses only describe the set of reachable intermediate states. It is possible to record the relationship to initial states also in the resolution diagram by adding ghost variables storing the initial program state.

Since symbolic execution trees can be large or infinite, different \emph{search strategies} are applied in symbolic execution to systematically traverse the tree and explore possible program behavior. Classical
strategies include breadth-first search (BFS),
depth-first search (DFS), and iterative deepening. Many other
strategies have been developed in the symbolic execution
literature, such as generational search~\citep{DBLP:conf/ndss/GodefroidLM08} and concolic
execution~\citep{DBLP:conf/pldi/GodefroidKS05}, as well as various heuristic-guided approaches such as \klee's coverage-optimized search~\citep{DBLP:conf/osdi/CadarDE08}. Such strategies translate to constrained resolution. It can easily be seen, however, that applying constrained resolution in DFS manner is \emph{incomplete} for non-linear clauses, whereas fair strategies like BFS are refutationally complete (\Cref{sec:forward}). Using fair saturation of a given set of CHCs in combination with subsumption criteria, unit resolution is moreover an effective method for proving that a set of CHCs is \emph{satisfiable} by constructing the least model of the clauses; in this way, constrained resolution can often be used to prove the correctness of programs, even in the presence of unbounded loops. This article surveys some of these calculi and strategies under the unified view of symbolic execution and resolution over CHCs.

Resolution can also proceed from goals.
Applying resolution backward from goals corresponds to
symbolically executing the program in reverse from the error
states, deriving proof obligations that must hold at earlier
program locations, i.e., \emph{backward symbolic execution}~\cite{GulwaniJuvekar2009,DBLP:conf/sas/ChalupaS21}. Backward symbolic execution can be modeled by constrained unit resolution, now applying the resolution rule to derive \emph{negative} unit clauses starting from the goals. As illustrated in \Cref{sec:motivating}, backward strategies are often able to show the satisfiability of CHCs in cases where forward strategies diverge. It is (again) easy to see, however, that constrained unit resolution applied in backward direction is \emph{incomplete} for non-linear clauses, even for fair application strategies. A more appropriate version of resolution is \emph{SLD resolution} known from logic programming~\citep{DBLP:conf/ifip/Kowalski74,DBLP:journals/jacm/AptE82}; in fact, forward resolution corresponds to
\emph{bottom-up} evaluation in logic programming, while backward
resolution corresponds to \emph{top-down} evaluation. We study backward resolution in \Cref{sec:backward}. 

\section{Forward Execution by Unit Hyper-Resolution}\label{sec:forward}

As mentioned in Section~\ref{sec:chc-symex}, Forward Execution (\texttt{FE}) is a kind of bottom-up evaluation and corresponds to symbolically executing a program, starting from the initial state.
In this section, we introduce the calculus of unit hyper-resolution with a loose notion of subsumption, define a fairness criterion for strategies and show that the resulting algorithm is still refutationally complete.
Moreover, we give examples of unsatisfiable sets of CHCs with the proofs of their unsatisfiability and a way to extract models from a set of facts that is saturated w.r.t.\ our subsumption criterion. 
\subsection{The Algorithm of \texttt{FE}}
The core of \texttt{FE} is the inference rule of unit hyper-resolution:
\begin{equation}\label{eq:hyper_resolution}
\resolve{
  \chc{H}{\varphi}{\bigwedge_{i=1}^n P_i(\overline{s_i})}
  \qquad
  \chcfact{P_1(\overline{t_1})}{\iota_1}
  \qquad\dots\qquad
  \chcfact{P_n(\overline{t_n})}{\iota_n}
}{
  \chcfact{H}{\varphi \land \bigwedge_{i=1}^n (\iota_i \land \overline{t_i} = \overline{s_i})}
}
\end{equation}
The inference rule is a combination of \emph{constrained resolution} (\Cref{sec:prelim}) with \emph{hyper-resolution} and guarantees that any consequence obtained by applying it is a fact.
It is well-known that unit hyper-resolution is refutationally complete if it is equipped with a \emph{fair} search strategy, meaning that all possible branches are explored and the strategy does not go down some branch indefinitely~\citep{DBLP:books/daglib/Harrison09}.
Thus, unit hyper-resolution is a suitable basis for proving the unsatisfiability of clause sets. 

However, we might apply the algorithm to some satisfiable set of clauses and thus, will never derive a contradiction.
In such a case, it would be nice if the algorithm terminated as well and gave us the result \texttt{SAT}.
Intuitively, what we need for this is a redundancy-criterion that tells us whether facts that can be derived with hyper-resolution are already entailed by the current set of facts.
We can stop deriving facts, when all facts that can be derived are redundant w.r.t.\ this criterion.
For propositional logic, there is a simple redundancy criterion: a fact is redundant if it is already an element of the clause set.
This works because for propositional facts, which are just propositional atoms, semantic entailment is equivalent to semantic equality, which in turn is equivalent to syntactic equality.
For first order-logic, this is much more subtle: consider the (semantically equivalent) facts $\chcfact{P(x)}{x \geq 0}$ and $\chcfact{P(x)}{x+x \geq 0}$.
Unfortunately, semantic entailment in first order logic is undecidable in general, so we cannot use it directly in our algorithm.
For our purposes, it suffices to assume that we have a procedure \texttt{subsumes} with the following properties:
\begin{enumerate}
    \item \texttt{subsumes} takes two arguments, a set of constrained Horn clauses $\calC$ and some constrained Horn clause $D$, and returns a Boolean value $\{\texttt{true}, \texttt{false}\}$;
    \item \texttt{subsumes} always terminates;
    \item If $D \in \calC$, then $\texttt{subsumes}(\calC, D) = \texttt{true}$;
    \item If $\texttt{subsumes}(\calC, D) = \texttt{true}$, then $\bigwedge_{C \in \calC} C \vDash D$.
\end{enumerate}
Informally, these properties of \texttt{subsumes} ensure that \texttt{subsumes} is a sound and decidable approximation of semantic entailment that, at the very least, recognizes syntactic equality/membership.
Beyond that, we do not make any assumptions about \texttt{subsumes} as different versions have different advantages.
In \Cref{sec:subsumption}, there is a detailed discussion about subsumption and we also give examples for such a procedure.

Now that we have all ingredients for our \texttt{FE}-algorithm, we still need a bit of notation before we can give the algorithm:
\begin{itemize}
    \item For a set of facts $\calF$ and a set of rules $\calR$, $\HR(\calR, \calF)$ denotes the set of hyper-resolvents $F$ obtained by resolving clauses of $\calC$ with facts of $\calF$
    \item For a set of facts $\calF$ and a set of rules $\calR$, $\HR_{\texttt{subsumes}}(\calR, \calF)$ denotes the set $\{F \in \HR(\calR, \calF) \mid \texttt{subsumes}(\calF, F) = \texttt{false}\}$
    \item For a set of goal clauses $\calG$ and a set of facts~$\calF$, $\Conf(\calG, \calF)$ denotes the set of hyper-resolvents $\chcfact{\bot}{\varphi(\overline{x})}$ obtained by resolving goals from $\calG$ with facts from $\calF$ s.t.\ $\vDash \exists\overline{x}. \varphi(\overline{x})$
\end{itemize}
We can now give the abstract transition system of \texttt{FE} for sets of CHCs $\calC = \calF \cup \calR \cup \calG$ given as the disjoint union of facts $\calF$, rules $\calR$ and goals $\calG$:
\begin{figure}[h!]
    \centering
    \begin{alignat*}{3}
        &\text{Unfold} \quad && \calF \mid \calR \mid \calG \Rightarrow \calF \cup \{F\} \mid \calR \mid \calG \quad&& \text{if $F \in \HR_{\texttt{subsumes}}(\calR, \calF)$} \\
        &\text{UNSAT} \quad &&\calF \mid \calR \mid \calG \Rightarrow \texttt{return UNSAT}\quad &&\text{if $\Conf(\calG, \calF) \neq \emptyset$} \\
        &\text{SAT} \quad &&\calF \mid \calR \mid \calG \Rightarrow \texttt{return SAT} \quad&&\text{if $\Conf(\calG, \calF) = \emptyset = \HR_{\texttt{subsumes}}(\calR, \calF)$}
    \end{alignat*}
    \caption{The Abstract Transition System for \texttt{FE}}
    \label{fig:forward_execution}
\end{figure}
\subsection{Refutational Completeness of \texttt{FE}}
The abstract transition system in Figure~\ref{fig:forward_execution} gives us a non-deterministic algorithm for \texttt{FE}, but leaves room for various strategies to decide which facts to add.
In this section, we are not interested in the precise strategy, but we do care about the strategy being \emph{fair}. 
\begin{definition}
    An execution of \texttt{FE} is \emph{fair} if for every state $\calF \mid \calR \mid \calG$ and $F \in \HR_{\texttt{subsumes}}(\calR, \calF)$, eventually some state $\calF \mid \calR \mid \calG'$ is reached, where $\texttt{subsumes}(\calF', F) = \texttt{true}$.
    A strategy for \texttt{FE} that ensures that every execution is fair is called a \emph{fair} strategy.
\end{definition}
Note that we can easily ensure fair executions by stratifying the clauses of $\HR_{\texttt{subsumes}}(\calR, \calF)$.

As mentioned above, it is a well-known fact that unit hyper-resolution equipped with a fair strategy is refutationally complete.
One might be inclined to use this as an argument for reasoning that \texttt{FE} equipped with a fair strategy is also refutationally complete.
While this is intuitive, it is a-priori not clear that we did not sacrifice refutational completeness by adding the redundancy criterion.
The following theorem, however, shows that our intuition is correct:
\begin{theorem}\label{theorem:main_fe}
    If \texttt{FE} is equipped with a fair strategy, then it is refutationally complete, i.e.\ \texttt{FE} terminates on $\calC$ with \texttt{UNSAT} if and only if $\calC$ is unsatisfiable.
\end{theorem}
\begin{proof}
    It is clear that if \texttt{FE} terminates with \texttt{UNSAT}, then $\calC$ is unsatisfiable.
    We only have to show that if $\calC$ is unsatisfiable, then \texttt{FE} terminates with \texttt{UNSAT} if it is run with a fair strategy.

    Assume that we are given a set of clauses $\calC = \calF \cup \calR \cup \calG$.
    It follows inductively that for every fact $F$ that can be derived from $\calC$ with hyper-resolution, \texttt{FE} reaches a state $\calF' \mid \calR \mid \calG$ s.t.\ $\texttt{subsumes}(\calF', F) = \texttt{true}$.
    Assume that $\calC$ is unsatisfiable.
    Since hyper-resolution is refutationally complete, there is a goal $\chc{\bot}{\bigwedge_{i=1}^n P_i(\overline{x_i})}{\phi(\overline{x_1}, \dots, \overline{x_n})}$ in $\calG$ and facts $F_1, \dots, F_n$ with $F_i : \chcfact{P_I(\overline{x_i})}{\iota_i(\overline{x_i})}$ that can be derived with hyper-resolution s.t.\ $\vDash \exists \overline{x_1}, \dots, \overline{x_n}\bigwedge_{i=1}^n \iota_i(\overline{x_i}) \land \phi(\overline{x_1}, \dots, \overline{x_n})$.
    There is a state $\calF^* \mid \calR \mid \calG$ s.t.\ $\texttt{subsumes}(\calF^*, F_i) = \texttt{true}$ for all $F_i$.
    Pick any $\overline{a_1}, \dots, \overline{a_n}$ s.t.\ $\vDash \bigwedge_{i=1}^n \iota_i(\overline{a_i}) \land \phi(\overline{a_1}, \dots, \overline{a_n})$.
    There are facts $F_i^* \in \calF^*$ s.t.\ $F_i^* : \chcfact{P_i(\overline{x_i})}{\iota_i^*(\overline{x_i})}$ and $\vDash \iota_i^*(\overline{a_i})$ for all $i$.
    In particular, $\vDash \exists \overline{x_1}, \dots, \overline{x_n}\bigwedge_{i=1}^n \iota_i^*(\overline{x_i}) \land \phi(\overline{x_1}, \dots, \overline{x_n})$.
    Thus, $\Conf(\calG, \calF^*)$ is not empty and \texttt{FE} terminates with \texttt{UNSAT}.
\end{proof}
\subsection{Constructing Models from \texttt{FE}}
\Cref{theorem:main_fe} implies that \texttt{FE} is sound, which means that a set of clauses $\calC$ is unsatisfiable if \texttt{FE} terminates with \texttt{UNSAT} and satisfiable if \texttt{FE} terminates with \texttt{SAT}.
Moreover, the last state of the algorithm before terminating with \texttt{SAT} allows us to extract a model of $\calC$ quite easily.
\begin{definition}
    Let $\calF$ be a finite set of facts and $P$ a predicate symbol. 
    We define $\calF_P = \{\iota(\overline{x}) \mid \chcfact{P(\overline{x})}{\iota(\overline{x})} \in \calF\}$. W.l.o.g., we assume that all $\iota \in\calF_P$ have the same free variables.
    We define $P(\calF) = \bigvee_{\iota \in \calF_P} \iota$.
\end{definition}
The following theorem is standard and thus, we state it without proof:
\begin{theorem}
    Let $\calC = \calF \cup \calR \cup \calG$ be a set of CHCs. If \texttt{FE} terminates with \texttt{SAT} on $\calC$ and the last state of \texttt{FE} before terminating is $\calF^* \mid \calR \mid \calG$, then the least model of $\calC$ is given by interpreting each $P$ with $P(\calF^*)$.
\end{theorem}

\section{Backward Execution by Hyper-SLD-Resolution}\label{sec:backward}

In \Cref{sec:forward}, we considered \texttt{FE}, which derives new facts, until the set of facts either stabilizes or an error is found. 
This corresponds to symbolically executing a program from the start.
Symbolic execution can also be performed in the other direction: we start with an initial set of error states and iteratively add the states that lead to an error in one step.
In the setting of Horn clauses, this corresponds to starting with an initial set of goals and adding new goals by resolving with rules.
This approach has been successfully employed in the unconstrained case as the calculus underlying Prolog.
However, we can easily adapt this to our setting of CHCs and define the constrained version of \emph{SLD-resolution:}
\subsection{The Algorithm of \texttt{BE}}
Similar to \texttt{FE} there is a single inference rule underlying \texttt{BE}.
Generalizing SLD-resolution from the unconstrained setting (e.g.\ Prolog) to CHCs yields the following, constrained, version of SLD:
\begin{equation}\label{eq:sld}
\resolve{
  \chc{\bot}{P(\overline{s})\land B(\overline{t})}{\varphi}
  \qquad
  \chc{P(\overline{q})}{B'(\overline{r})}{\psi}
}{
  \chc{\bot}{B(\overline{t})\land B'(\overline{r})}{\varphi \land \psi \land \overline{q} = \overline{s}},
}
\end{equation}
Similarly to \Cref{sec:forward}, we will add a redundancy criterion, so that we can deal with satisfiable clause-sets.
However, in order to prove that the resulting algorithm is still refutationally complete, we need a slightly more general version of SLD-resolution, which we call hyper-SLD-resolution:
\begin{equation}\label{eq:hsld}
\resolve{
  \chc{\bot}{\bigwedge_{i=1}^n P_i(\overline{s_i}) \land B}{\varphi}
  \qquad
  \chc{P_1(\overline{t_1})}{A_1}{\psi_1}
  \qquad \dots \qquad
  \chc{P_n(\overline{t_n})}{A_n}{\psi_n}
}{
  \chc{\bot}{B\land \bigwedge_{i=1}^n A_i}{\varphi \land \bigwedge_{i=1}^n (\psi_i \land \overline{s_i} = \overline{t_i})},
}
\end{equation}
Note that the SLD-resolution rule is a special instance of hyper-SLD-resolution with $n=1$.
On the other hand, every hyper-SLD-derivation can easily be transformed into an SLD-derivation.
Note that, in contrast to hyper-resolution, hyper-SLD-resolution does not have to yield a unit clause, nor do we have to resolve all atoms of the body.
Hyper-SLD-resolution simply allows us to perform several SLD-resolution simultaneously, but in a very controlled way.

Similarly to hyper-resolution, SLD resolution is known to be refutationally complete if it is equipped with a \emph{fair} strategy~\citep{DBLP:books/daglib/Harrison09}.
In particular, hyper-SLD-resolution is refutationally complete and therefore, it is a suitable basis for proving the unsatisfiability of clause sets.
In order to add a redundancy criterion, we assume to have a procedure \texttt{subsumes} with the same properties as in \Cref{sec:forward}.
Before giving the algorithm, we need a little bit of notation:
\begin{itemize}
    \item For a set of goal clauses $\calG$ and a set of rules $\calR$, $HSLD(\calR, \calG)$ denotes the set of resolvents $G$ that are obtained by applying hyper-SLD resolution to goals of $\calG$ and rules of $\calR$
    \item For a set of goal clauses $\calG$ and a set of rules $\calR$, $HSLD_{\texttt{subsumes}}(\calR, \calG)$ denotes the set $\{G \in HSLD(\calR, \calG) \mid \texttt{subsumes}(\calG,G) = \texttt{false}\}$
    \item For a set of goal clauses $\calG$ and a set of facts $\calC$, $\Conf(\calG, \calF)$ is defined as in \cref{sec:forward}
\end{itemize}
We can now give the abstract transition system for \texttt{BE} for a set of CHCs $\calC = \calF \cup \calR \cup \calG$ given as the disjoint union of facts $\calF$, rules $\calR$ and goals $\calG$:
\begin{figure}[h!]
    \centering
    \begin{alignat*}{3}
        &\text{Unfold} \quad && \calF \mid \calR \mid \calG \Rightarrow \calF \mid \calR \mid \calG\cup \{G\} \quad&& \text{if $G \in HSLD_{\texttt{subsumes}}(\calR, \calG)$} \\
        &\text{UNSAT} \quad &&\calF \mid \calR \mid \calG \Rightarrow \texttt{return UNSAT}\quad &&\text{if $\Conf(\calG, \calF) \neq \emptyset$} \\
        &\text{SAT} \quad &&\calF \mid \calR \mid \calG \Rightarrow \texttt{return SAT} \quad&&\text{if $\Conf(\calG, \calF) = \emptyset = HSLD_{\texttt{subsumes}}(\calR, \calG)$}
    \end{alignat*}
    \caption{The Abstract Transition System for \texttt{BE}}
    \label{fig:backward_execution}
\end{figure}

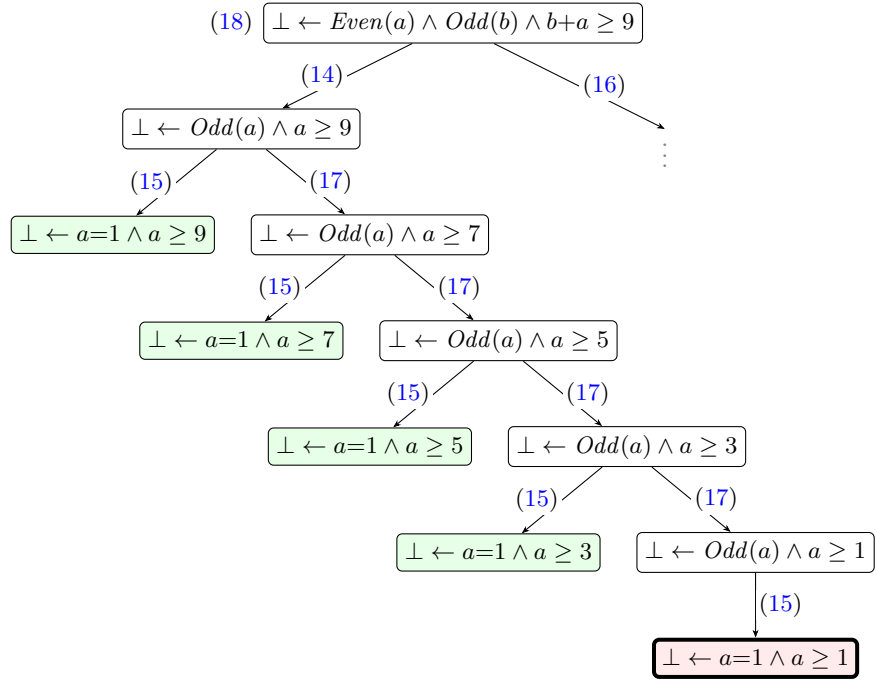
\begin{figure}[t]
\begin{center}
\begin{tikzpicture}[every node/.append style={font=\scriptsize}]
  \node[nd] (n0) at (2.8,0)
    {$\false \cmark \mathit{Even}(a) \csep \mathit{Odd}(b) \csep b {+} a \geq 9$};
  \node[cl] at ([xshift=-1.2em]n0.west) {\eqref{cl:eog}};

  \node[nd] (n1) at (0,-1.4)
    {$\false \cmark \mathit{Odd}(a) \csep a \geq 9$};

  \draw[arr] (n0) -- (n1) node[cl,midway,left] {\eqref{cl:eb}};
  \draw[arr] (n0) -- ++(2.8,-1.4) coordinate (rtip)
    node[cl,midway,right] {\eqref{cl:er}};
  \node[font=\scriptsize, text=gray] at ([yshift=-0.7em]rtip) {$\vdots$};

  \node[ok] (n3) at (-1.7,-2.8)
    {$\false \cmark a {=} 1 \csep a \geq 9$};
  \node[nd] (n4) at (1.7,-2.8)
    {$\false \cmark \mathit{Odd}(a) \csep a \geq 7$};

  \draw[arr] (n1) -- (n3) node[cl,midway,left] {\eqref{cl:ob}};
  \draw[arr] (n1) -- (n4) node[cl,midway,right] {\eqref{cl:or}};

  \node[ok] (n7) at (0,-4.2)
    {$\false \cmark a {=} 1 \csep a \geq 7$};
  \node[nd] (n8) at (3.4,-4.2)
    {$\false \cmark \mathit{Odd}(a) \csep a \geq 5$};

  \draw[arr] (n4) -- (n7) node[cl,midway,left] {\eqref{cl:ob}};
  \draw[arr] (n4) -- (n8) node[cl,midway,right] {\eqref{cl:or}};

  \node[ok] (n13) at (1.7,-5.6)
    {$\false \cmark a {=} 1 \csep a \geq 5$};
  \node[nd] (n14) at (5.1,-5.6)
    {$\false \cmark \mathit{Odd}(a) \csep a \geq 3$};

  \draw[arr] (n8) -- (n13) node[cl,midway,left] {\eqref{cl:ob}};
  \draw[arr] (n8) -- (n14) node[cl,midway,right] {\eqref{cl:or}};

  \node[ok] (n19) at (3.4,-7.0)
    {$\false \cmark a {=} 1 \csep a \geq 3$};
  \node[nd] (n20) at (6.8,-7.0)
    {$\false \cmark \mathit{Odd}(a) \csep a \geq 1$};

  \draw[arr] (n14) -- (n19) node[cl,midway,left] {\eqref{cl:ob}};
  \draw[arr] (n14) -- (n20) node[cl,midway,right] {\eqref{cl:or}};

  \node[hi] (n25) at (6.8,-8.4)
    {$\false \cmark a {=} 1 \csep a \geq 1$};

  \draw[arr] (n20) -- (n25) node[cl,midway,right] {\eqref{cl:ob}};
\end{tikzpicture}
\end{center}
\caption{Backward resolution for \Cref{ex:even-odd}.}
\label{fig:backward-nonlinear}
\end{figure}

\begin{example}\label{ex:even-odd}
Consider the following CHCs that define the non-negative even and odd numbers with the non-linear goal~\eqref{cl:eog}:
\begin{align}
  \achcfact{\mathit{Even}(a)}{a = 0}
    \label{cl:eb} \\[2pt]
  \achcfact{\mathit{Odd}(a)}{a = 1}
    \label{cl:ob} \\[2pt]
  \achc{\mathit{Even}(a)}{\mathit{Even}(b)}{a = b + 2}
    \label{cl:er} \\[2pt]
  \achc{\mathit{Odd}(a)}{\mathit{Odd}(b)}{a = b + 2}
    \label{cl:or} \\[2pt]
  \achcgoal{\mathit{Even}(a) \csep \mathit{Odd}(b)}{b + a \geq 9}
    \label{cl:eog}
\end{align}
The clause set is unsatisfiable, since the sum in the goal will eventually exceed $9$.
\Cref{fig:backward-nonlinear} shows the backward resolution tree. Resolving $\mathit{Even}(a)$ in~\eqref{cl:eog} with~\eqref{cl:eb} eliminates $\mathit{Even}$, yielding the linear goal $\false \cmark \mathit{Odd}(a) \csep a \geq 9$. The branch then continues resolving with \eqref{cl:or}, eventually deriving a contradiction. The right branch from~\eqref{cl:er} is also explored but not shown.
\end{example}

\subsection{Refutational Completeness of \texttt{BE}}
Similarly to \Cref{sec:forward} the abstract transition system from \Cref{fig:backward_execution} leaves some freedom w.r.t. the strategy for which goals to add.
Again, the concrete strategy does not matter, only \emph{fairness}:
\begin{definition}
    An execution of \texttt{BE} is \emph{fair} if for every state $\calF \mid \calR \mid \calG$ and $G \in HSLD_{\texttt{subsumes}}(\calR, \calG)$, eventually some state $\calF \mid \calR \mid \calG'$ is reached, where $\texttt{subsumes}(\calG', G) = \texttt{true}$.
    A strategy for \texttt{BE} that ensures that every execution is fair is called a \emph{fair} strategy.
\end{definition}
Note that we can easily ensure fair executions by stratifying the clauses of $HSLD_{\texttt{subsumes}}(\calR, \calG)$.

Similarly to \Cref{sec:forward}, one might find it intuitive that \texttt{BE} is refutationally complete because SLD-resolution is.
Once again, this intuition will turn out to be correct, but goal clauses, in particular non-linear ones, are more peculiar than facts, so this is not as trivial as it was for \texttt{FE}.

First, we need to think about sets of goals and consequences of SLD-resolution.
Interestingly enough, SLD-resolution is not monotone w.r.t. logical entailment, as the following example shows:
\begin{example}
    Assume we have one rule $R: \chcfact{P(x)}{Q(x)}$ and the following two goals $G_1 : \chc{\bot}{P(x)\land P(y)}{x\geq y}$ and $G_2 : \chcfact{\bot}{P(x)}$. 
    It holds that $\vDash G_1\leftrightarrow G_2$ and in particular, $G_1 \vDash G_2$.
    Let $S_i$ be the set of all possible SLD-resolvents of $G_i$ and $R$.
    We have that $S_1 = \{\chc{\bot}{P(x) \land Q(y)}{x \geq y}, \chc{\bot}{Q(x) \land P(y)}{x \geq y}\}$ and $S_2 = \{\chcfact{\bot}{Q(x)}\}$.
    It follows that $G_1, S_1 \nvDash \bigwedge_{G \in S_2} G$. 
\end{example}
The example above might give some idea, why we failed working with SLD-resolution and introduced hyper-SLD-resolution.
The following lemma shows that hyper-SLD-resolution is monotone w.r.t. logical entailment:
\begin{lemma}\label{lemma:hsld_monotonicity}
    Let $\calG_1$ and $\calG_2$ be sets of goal clauses and $\calR$ be a set of rules. If $\calG_1 \vDash G$ for any $G \in \calG_2$, then $\calG_1, HSLD(\calR, \calG_1) \vDash G$ for any $G \in HSLD(\calR, \calG)$.
    In particular, $\calG_1, HSLD_{\texttt{subsumes}}(\calR, \calG_1) \vDash G$ for any $G \in HSLD(\calR, \calG)$.
\end{lemma}
\begin{proof}
    Take any goal $G : \chc{\bot}{B(\overline{x}) \land \bigwedge_{i=1}^n A_i(\overline{y_i})}{\psi(\overline{x}, \overline{z_1},\dots, \overline{z_n}) \land \bigwedge_{i=1}^n \tau_i(\overline{y_i}, \overline{z_i})}$ in $HSLD(\calR, \calG_2)$.
    There is a goal $G' : \chc{\bot}{B(\overline{x}) \land \bigwedge_{i=1}^n P_i(\overline{z_i})}{\psi(\overline{x}, \overline{z_1}, \dots, \overline{z_n})} \in \calG_2$ and rules $R_1, \dots, R_n \in \calR$ with $R_i : \chc{P_i(\overline{z_i})}{A_i(\overline{y_i})}{\tau_i(\overline{y_i}, \overline{z_i})}$.
    Pick any model $M$ of the background theory $T$ that interprets the uninterpreted predicate symbols as well.
    We must show that for any $\overline{a}$ and $\overline{b_1},\overline{c_1},\dots, \overline{b_n}, \overline{c_n}$, there is a goal $\chc{\bot}{B^*(\overline{u})}{\phi(\overline{u})}$ in $\calG_1 \cup HSLD(\calR, \calG_1)$ s.t.\ $M \vDash\left(B(\overline{a}) \land \bigwedge_{i=1}^n A_i(\overline{b_i}) \land \psi(\overline{a}, \overline{c_1},\dots, \overline{c_n}) \land \bigwedge_{i=1}^n \tau_i(\overline{b_i}, \overline{c_i}) \right)\rightarrow \left(\exists \overline{u}. B^*(\overline{u}) \land \phi(\overline{u})\right)$.

    If $M \vDash \neg (B(\overline{a}) \land \bigwedge_{i=1}^n A_i(\overline{b_i}) \land \psi(\overline{a}, \overline{c_1},\dots, \overline{c_n}) \land \bigwedge_{i=1}^n \tau_i(\overline{b_i}, \overline{c_i}))$, then we are done.

    Assume that $M \vDash B(\overline{a}) \land \bigwedge_{i=1}^n A_i(\overline{b_i}) \land \psi(\overline{a}, \overline{c_1},\dots, \overline{c_n}) \land \bigwedge_{i=1}^n \tau_i(\overline{b_i}, \overline{c_i})$.
    Since $\calG_1 \vDash G'$, there is a goal $G'' : \chc{\bot}{B''(\overline{u})}{\phi''(\overline{u})}\in \calG_1$ with $M \vDash B(\overline{a}) \land \bigwedge_{i=1}^n P_i(\overline{c_i}) \land \psi(\overline{a},\overline{c_1}, \dots, \overline{c_n}) \rightarrow \exists \overline{u}.B''(\overline{u}) \land \phi''(\overline{u})$.
    Since $B''$ consists of uninterpreted predicate symbols only, it follows that the witness of $\overline{u}$ uses only elements of $\overline{a}, \overline{c_1}, \dots, \overline{c_n}$.
    Thus, w.l.o.g., $G''$ has the form $\chc{\bot}{B''(\overline{v}, \overline{u_1}, \dots, \overline{u_n})}{\phi''(\overline{v}, \overline{u_1}, \dots, \overline{u_n})}$ and $M \vDash B(\overline{a}) \land \bigwedge_{i=1}^n P_i(\overline{c_i}) \land \psi(\overline{a},\overline{c_1}, \dots, \overline{c_n}) \rightarrow B''(\overline{a}, \overline{c_1}, \dots, \overline{c_n}) \land \phi''(\overline{a}, \overline{c_1}, \dots, \overline{c_n})$.
    In particular, $M \vDash \psi(\overline{a}, \overline{c_1}, \dots, \overline{c_n}) \rightarrow \phi(\overline{a}, \overline{c_1}, \dots, \overline{c_n})$.
    We are almost done. 
    The only problem is that $B''$ might contain the predicate symbol $P_i$ s.t.\ $B''(\overline{a}, \overline{c_1},\dots, \overline{c_n})$ contains $P_i(\overline{c_i})$.
    However, we can use one application of the hyper-SLD-rule to cut all such atoms with the respective rule $R_i$.
    Note that we do not necessarily cut all atoms containing $P_i(\overline{x})$, but only those where $\overline{x}$ is substituted with $\overline{c_i}$.
    After applying this resolution, we obtain a goal $\chc{\bot}{B^*(\overline{u})}{\phi(\overline{u}, \overline{v_1},\dots, \overline{v_n}) \land \bigwedge_{i=1}^n \tau_i(\overline{w_i}, \overline{v_i})}$ with $M \vDash \left(B(\overline{a}) \land \bigwedge_{i=1}^n P_i(\overline{c_i}) \land \psi(\overline{a},\overline{c_1}, \dots, \overline{c_n})\right) \rightarrow \left(B^*(\overline{a}, \overline{c_1},\dots,\overline{c_n}) \land \phi(\overline{a}, \overline{c_1},\dots, \overline{c_n}) \land \bigwedge_{i=1}^n \tau_i(\overline{b_i}, \overline{c_i})\right)$.
\end{proof}
\begin{theorem}
    If \texttt{BE} is run with a fair strategy, then it is refutationally complete, i.e.\ \texttt{BE} terminates with \texttt{UNSAT} on input $\calC$ if and only if $\calC$ is unsatisfiable.
    \end{theorem}
\begin{proof}
    It is clear that if \texttt{BE} terminates with \texttt{UNSAT}, then $\calC$ is unsatisfiable. 
    We only need to show that if a clause set is unsatisfiable, then \texttt{BE} terminates with \texttt{UNSAT}.

    Assume that $\calC$ is unsatisfiable and define the sets $\calG_i$ inductively: $\calG_i = \calG$, $\calG_{i+1} = \calG_i \cup HSLD(\calR, \calG_i)$.
    There is some $i_0$, a goal $G: \chc{\bot}{\bigwedge_{k=1}^n P_k(\overline{x_k})}{\phi(\overline{x_1},\dots, \overline{x_n})}$ in $\calG_{i_0}$ and facts $F_k : \chcfact{P_k(\overline{x_k})}{\iota_k(\overline{x_k})}$ in $\calF$ s.t.\ $\vDash \exists \overline{x_1},\dots, \overline{x_n}.(\bigwedge_{k=1}^n\iota_k(\overline{x_k} \land \phi(\overline{x_1},\dots, \overline{x_n}))$.
    From Lemma~\ref{lemma:hsld_monotonicity} it follows that every fair execution of \texttt{BE} reaches a state $\calF \mid \calR \mid \calG^*$ s.t.\ $\calG^* \vDash G$.
    Pick any model $M$ and any $\overline{a_1},\dots, \overline{a_k}$ s.t.\ $M \vDash \bigwedge_{k=1}^n\iota_k(\overline{a_k}) \land \phi(\overline{a_1},\dots, \overline{a_n})$.
    There is a goal $G^* : \chc{\bot}{\bigwedge_{k=1}^m Q_k(\overline{y_k})}{\psi(\overline{y_1}, \dots, \overline{y_m})}$ in $\calG^*$ s.t.\ $M \vDash (\bigwedge_{k=1}^n P_k(\overline{a_k}) \land \phi(\overline{a_1},\dots, \overline{a_n})) \rightarrow \exists \overline{y_1}, \dots, \overline{y_m}.\bigwedge_{k=1}^mQ_k(\overline{y_k}) \land \psi(\overline{y_1},\dots, \overline{y_m})$.
    It follows that every $Q_k$ is some $P_l$ and that there is some substitution $\theta$ mapping $\overline{y_k}$ to $\overline{x_l}$ s.t.\ $\psi^*(\overline{x_1},\dots, \overline{x_n}) \equiv \psi\theta$ has the property that $M \vDash \phi(\overline{a_1}, \dots, \overline{a_n}) \rightarrow \psi^*(\overline{a_1}, \dots, \overline{a_n})$.
    By hyper-resolving every $Q_k$ with the appropriate $P_l$, we obtain the clause $\chcfact{\bot}{\bigwedge_{i=1}^k \iota_k(\overline{x_k})\land \psi(\overline{y_1}, \dots, \overline{y_m}) \land \bigwedge_{k=1}^m \overline{y_k}} = \overline{y_k}\theta$.
    It follows that $M \vDash \exists \overline{x_1},\dots ,\overline{x_n},\overline{y_1},\dots, \overline{y_m}.(\bigwedge_{i=1}^k \iota_k(\overline{x_k})\land \psi(\overline{y_1}, \dots, \overline{y_m}) \land \bigwedge_{k=1}^m \overline{y_k} = \overline{y_k}\theta)$. 
    Since $M$ was arbitrary, \texttt{BE} terminates with \texttt{UNSAT}.
\end{proof}
\subsection{Constructing Models from \texttt{BE}}
In \cref{sec:forward}, we were able to give a nice method to construct models from saturated sets of facts after executing \texttt{FE}.
Maybe surprisingly, this is much more intricate for \texttt{BE}.
The two main factors contributing to this are non-linearity of goals and the potential non-existence of greatest fixed points.
Assume that we are given a set of CHCs $\calC = \calF \cup \calR \cup \calG$.

One view of \texttt{FE} is that we under-approximate the set of states that is necessarily contained in every solution of $\calC$ by deriving facts until the set of facts stabilizes, meaning that we reached the least fixed point of $\calF \cup \calR$.
If we have not derived a contradiction in the process, then the least fixed point of $\calF \cup \calR$ is a model of $\calC$.
Moreover, we can write it down explicitly as the disjunction of the facts. 
This is possible because the facts are linear and therefore do not display any dependencies between different predicate symbols.

If we instead use \texttt{BE}, then we over-approximate the set of states that are necessarily not contained in any solution of of $\calC$ by deriving new goals.
Intuitively, we aim for the greatest fixed point of $\calC$.
However, there does not have to be a greatest fixed point of $\calC$ and, more importantly, \texttt{BE} can also terminate with \texttt{SAT} on a clause set if it does not possess a greatest fixed point.  
This begs the question, which fixed point we would even like to define.
The obstacle that prevents us from constructing \emph{any} solution explicitly is non-linearity: A goal like $\chcfact{\bot}{P(x) \land Q(x)}$ tells us that $P$ and $Q$ are disjoint in any model.
In particular, the interpretations of $P$ and $Q$ depend on each other and we cannot define them separately as it was the case with \texttt{FE}.
Nonetheless, if the saturated set of goals only contains linear clauses, then we can extract a model in a way that is dual to what we did with \texttt{FE}.
\begin{definition}
    Let $\calG$ be a finite set of linear goals and $P$ a predicate symbol. 
    We define $\calG_P = \{\phi(\overline{x}) \mid \chc{\bot}{P(\overline{x})}{\phi(\overline{x})} \in \calG\}$. W.l.o.g., we assume that all $\phi \in\calG_P$ have the same free variables.
    We define $P(\calG) = \bigwedge_{\phi \in \calG_P} \neg\phi$.
\end{definition}
The following theorem is standard and thus, we state it without proof:
\begin{theorem}
    Let $\calC = \calF \cup \calR \cup \calG$ be a set of CHCs. If \texttt{BE} terminates with \texttt{SAT} on $\calC$ and the last state of \texttt{BE} before terminating is $\calF \mid \calR \mid \calG^*$ s.t.\ every goal in $\calG$ is a linear clause, then the greatest model of $\calC$ is given by interpreting each $P$ with $P(\calG^*)$.
\end{theorem}
The theorem above only deals with the case of linear goals.
We leave the question of model construction in the general case for future research.

\section{Backward Execution vs.\ $k$-Induction}\label{sec:k-induction}
So far, we have only considered CHCs as representation for programs and imposed no syntactic restrictions on them to run \texttt{FE} or \texttt{BE}.
However, there is a plethora of other ways to represent programs, some of which allow us to employ methods to prove their safety that might not be suited for arbitrary sets of CHCs.
In this section, we present $k$-induction for a restricted fragment of CHC-sets and show that \texttt{BE} can be considered to be a natural extension of $k$-induction in the setting of arbitrary sets of CHCs.

The notion $k$-induction was first introduced in~\cite{DBLP:conf/fmcad/SheeranSS00} for transition systems. Transition systems can easily be translated to set of CHCs, but we can also simply define $k$-induction for sets of CHCs directly if we restrict the sets of CHCs that we consider. Assume that we are given a set of CHCs $\calC$ of the following form:
\begin{equation}\label{eq:definition_kind}
    \{\chcfact{P(\overline{x})}{\iota(\overline{x})}, \ \chc{P(\overline{y})}{P(\overline{x})}{\tau(\overline{x}, \overline{y})}, \ \chc{\bot}{P(\overline{x})}{\beta(\overline{x})} \}.
\end{equation}
Note that we can translate any transition system into a set of the form \eqref{eq:definition_kind}.
We can now verify that $\calC$ is safe by finding some $k > 0$ s.t.\ the following two statements hold:
\begin{description}
    \item[Base case] All paths of length at most $k$ starting from an initial state are safe.
    \item[Induction step] Every path of length $k+1$ whose prefix of length $k$ is safe is also safe.
\end{description}
More formally, these two statements correspond to the following logical formulas:
\begin{align}
    &\vDash \left(\iota(\overline{x_1}) \land \bigwedge_{i=1}^{k-1} \tau(\overline{x_i}, \overline{x_{i+1}}) \right) \rightarrow \bigwedge_{i=1}^k \neg \beta(\overline{x_i})\label{eq:kind_base_case} \\
    &\vDash \left(\bigwedge_{i=1}^k \neg \beta(\overline{x_i}) \land \bigwedge_{i=1}^{k} \tau(\overline{x_i}, \overline{x_{i+1}})\right) \rightarrow \neg \beta(\overline{x_{k+1}}) \label{eq:kind_step}
\end{align}
Evidently, $k$-induction is sound: If no error-state is reachable in $k$ steps and whenever we have a chain of $k$ safe states, every state we add to the chain is safe as well, then every reachable state is safe.
The motivating example in \Cref{fig:backward} illustrates this with $k=2$.

While this formulation is straightforward for clause sets that have the form of~\eqref{eq:definition_kind}, it is difficult to define a similar satisfiability criterion for more general cases.
Although non-linearity in rules can be dealt with, there is no formulation of $k$-induction for non-linear goals that we are aware of.

Let us now go back to \texttt{BE}: Applying \texttt{BE} to $\calC$ for $k$-steps yields a set of goals of the form $\calG_k = \{\chc{\bot}{P(\overline{x_1})}{\bigwedge_{i=1}^{l-1} \tau(\overline{x_i}, \overline{x_{i+1}}) \land \beta(\overline{x_l})} \mid l \leq k\}$.
Note the following: Checking that $\Conf(\calG_k, \calF) = \emptyset$ corresponds precisely to checking that~\eqref{eq:kind_base_case} holds.
Moreover, if $HSLD_{\texttt{subsumes}}(\calR, \calG_k) = \emptyset$, then this implies that~\eqref{eq:kind_step} holds.
If we require that \texttt{subsumes} is not only sound, but also complete for this special fragment of goal clauses, then the converse holds: If~\eqref{eq:kind_step} holds, then $HSLD_{\texttt{subsumes}}(\calR, \calG_k) = \emptyset$.
Thus, in this case, we have a precise match between $k$-induction and \texttt{BE}, which justifies our view to consider \texttt{BE} as a generalization of $k$-induction in the case or arbitrary Horn clauses.

Note that this connection between $k$-induction and \texttt{BE} is not new and was already observed in the paper introducing $k$-induction~\cite{DBLP:conf/fmcad/SheeranSS00}.
A similar observation was made by Chalupa and Strej\v{c}ek~\citep{DBLP:conf/sas/ChalupaS21} for backward symbolic execution, and by Bj{\o}rner et al.~\citep{BjornerGMR15} who recast $k$-induction as a fold-unfold transformation on linear Horn clauses. 
However, classical $k$-induction and these formulations are restricted to linear clauses. Backward resolution is not: it applies to arbitrary sets of CHCs, including non-linear ones.

\section{Subsumption}\label{sec:subsumption}
In \Cref{sec:forward} and \Cref{sec:backward}, we have assumed to have a procedure \texttt{subsumes} that is a sound and decidable approximation of logical entailment between a finite set of clauses and another clause.
The reason we did this is twofold: On the one hand, we would like to be able to check for logical entailment, in order to add some redundancy criterion to our algorithms and be able to derive that some clause set is satisfiable.
On the other hand, we cannot work with logical entailment directly as first-order logic is famously undecidable.
In this section, we will give examples of decidable subsumption criteria with varying strength, meaning that some are better approximations of logical entailment than others at the expense of being more complicated. Each of these criteria naturally leads to a procedure \texttt{subsumes} with the desired properties from \Cref{sec:forward}. 
Moreover, we will see that there are notions of subsumption that are not only sound, but also complete if we only consider (sets of) facts or (sets of) goals.

Starting from syntactic equality, the first non-trivial version of subsumption that we consider is commonly known as $\theta$-subsumption:
\begin{definition}
    Let $C_1$ and $C_2$ be two constrained Horn clauses with $C_i: \chc{H_i}{B_i}{\phi_i}$. $C_1$ $\theta$-\emph{subsumes} $C_2$ if there is a variable substitution $\theta$ with the following properties:
    \begin{enumerate}
        \item $H_1 \theta \supseteq H_2$
        \item $B_1 \theta \subseteq B_2$
        \item $\vDash \phi_2 \rightarrow \phi_1\theta$
    \end{enumerate}
\end{definition}
Note that $\theta$-subsumption only deals with the question of whether one clause $C_1$ entails another clause $C_2$.
We can easily lift it to our setting, where we ask the question if a whole clause set $\calC$ entails some clause $C$, by checking if there is any clause $C'$ in $\calC$ s.t.\ $C'$ $\theta$-subsumes $C$. 
Clearly, if $\calC$ $\theta$-subsumes $C$, then $\calC \vDash C$.
However, the converse is not necessarily the case:
\begin{example}\label{example:failed_subsumption}
    Consider the clauses  $C_1 : \chc{\bot}{E(z)}{z \% 2 = 1}$ and $C_2 : \chc{\bot}{E(x)\land E(y)}{(x+y) \% 2 = 1}$.
    Then, $C_1 \vDash C_2$, but $C_1$ does not subsume $C_2$ because for no substitution $\theta$ does it hold that $(x+y) \% 2 = 1 \vDash (z \% 2 = 1)\theta$.
    Informally, one could say that if the sum of two numbers is odd, then one of summands is odd as well, but you do not know which one beforehand.
\end{example}
The example above illustrates the potential shortcomings of plain subsumption: It is uniform in the sense that there has to be one $\theta$ that works for all cases.
We can use this as a motivation to define more general versions of subsumption that cover more cases, where uniformity is too much to ask for:
\begin{definition}
    Let $C_1$ and $C_2$ be constrained Horn clauses with $C_i: \chc{H_i}{B_i}{\phi_i}$. A substitution $\theta$ is $C_1$-$C_2$-feasible if the following holds
    \begin{enumerate}
        \item $H_1 \theta \supseteq H_2$
        \item $B_1 \theta \subseteq B_2$
    \end{enumerate}
    Let $\Theta(C_1,C_2)$ be the set of all $C_1$-$C_2$-feasible substitutions.
    $C_1$ \emph{disjunctively subsumes} $C_2$ if it holds that 
    \[
        \vDash \phi_2 \rightarrow \bigvee_{\theta \in \Theta(C_1,C_2)} \phi_1 \theta
    \]
\end{definition}
Note that $\Theta(C_1, C_2)$ is finite and the formulas above are well-formed.
Again, we can lift disjunctive subsumption to the our setting by checking it for each clause in the set and it holds that if $\calC$ disjunctively subsumes $C$, then $\calC \vDash C$, meaning that disjunctive subsumption is sound.
Moreover, the notion disjunctive subsumption is strictly stronger than regular subsumption as the Example~\ref{example:failed_subsumption} shows.
However, there are still some cases, where a clause $C$ is entailed by a clause set $\calC$, but $\calC$ does not disjunctively subsume $C$. Moreover, there are cases that seem easy enough that we would want our notion of subsumption to hold. 
Consider the following example:
\begin{example}\label{example:failed_disjunctive_subsumption}
    Consider the clauses $C_1 : \chc{P(x)}{Q(x)}{0 \leq x \leq 7}$, $C_2 : \chc{P(y)}{Q(y)}{2 \leq y \leq 10}$ and $C_3: \chc{P(z)}{Q(z)}{0 \leq z \leq 10}$.
    It holds that $C_1\land C_2 \vDash C_3$, but neither $C_1$ nor $C_2$ (disjunctively) subsume $C_3$.
\end{example}
In the example above, the problem is not uniformity, but that the conjunction of the clauses $C_1$ and $C_2$ contains more information than $C_1$ and $C_2$ separately and both pieces of information are necessary to prove the entailment.
We can easily fix this by extending our notion of disjunctive subsumption: 
\begin{definition}
    Let $C^* : \chc{H^*}{B^*}{\phi^*}$ be a constrained Horn clause and $\calC$ a finite set constrained Horn clauses with $C: \chc{H_C}{B_C}{\phi_C}$ for $C$ in $\calC$. $\calC$ disjunctively subsume $C^*$ if it holds that
    \[
        \vDash \phi^* \rightarrow \bigvee_{C\in\calC} \bigvee_{\theta \in \Theta(C, C^*)} \phi_C\theta
    \]
\end{definition}
Note that this definition coincides with the previous definition of disjunctive subsumption if $n=1$.
Thus, we replace the previous definition of disjunctive subsumption with this more general one.
Disjunctive subsumption is still sound and it captures very natural cases like Example~\ref{example:failed_disjunctive_subsumption}.
However, there are still examples of clause sets that logically entail another clause, but where disjunctive subsumption fails:
\begin{example}
    Consider the clauses $C_1 : \chcfact{P(x)}{x = 0}$, $C_2 : \chcfact{Q(y)}{P(y)}$ and $C_3 : \chcfact{Q(x)}{x = 0}$. Then $C_1 \land C_2 \vDash C_3$, but $\{C_1, C_2\}$ does not disjunctively subsume $C_3$ since $\Theta(C_1, C_3) = \Theta(C_2,C_3) = \emptyset$ and the constraint of $C_3$ is satisfiable. 
\end{example}
The example above shows that any version of subsumption defined so far does not fully capture logical entailment.
Moreover, as mentioned above, any attempt to fully capture logical entailment while staying decidable is doomed to fail.
Thus, we will not try to refine our notion subsumption anymore. 
We close this section by proving that disjunctive subsumption as defined above, while generally for all clause sets, is complete if we restrict the setting to goals or facts: 
\begin{proposition}
    Let $\calF$ be finite set of facts and $F^*$ some fact. It holds that $\calF \vDash F^*$ if and only if $\calF$ disjunctively subsumes $F^*$.
\end{proposition}
\begin{proof}
    Assume that $\calF \vDash F^*$ and that $F^*$ has the form $\chcfact{P(\overline{x})}{\iota^*(\overline{x})}$.
    There are two cases: If $\vDash \forall\overline{x}.\neg \iota^*(\overline{x})$, then $\calF$ trivially disjunctively subsumes $F^*$.
    Otherwise, we can assume w.l.o.g., that all $F \in \calF$ have the form $\chcfact{P(\overline{x_F})}{\iota_F(\overline{x_F}})$, i.e.\ they all have the same predicate in the head.
    Since $\calF \vDash F^*$, it holds that for any model $M$ and any $\overline{a}$, there is some $F \in \calF$ with $M \vDash \iota^*(\overline{a}) \rightarrow \iota_F(\overline{a})$.
    Since $M$ and $\overline{a}$ were arbitrary, it follows that $\vDash \iota^*(\overline{x}) \rightarrow \bigvee_{F \in \calF} \iota_F(\overline{x})$.
    By defining the substitution $\theta_F = \{\overline{x_F} \mapsto \overline{x}\}$, this is equivalent to $\vDash \iota^*\rightarrow \bigvee_{F \in \calF} \iota_F\theta_F$, which fits our definition of disjunctive subsumption.
\end{proof}
\begin{proposition}
    Let $\calG$ be a finite set of goal clauses and $\calG^*$ some goal clause. It holds that $\calG \vDash G^*$ if and only if $\calG$ disjunctively subsume $G^*$.
\end{proposition}
\begin{proof}
    Assume that every goal $G$ in $\calG$ has the form $\chc{\bot}{B_G}{\psi_G}.$
    Assume that $\calG \vDash G^*$ and that $G^*$ has the form $\chc{\bot}{B^*(\overline{x})}{\phi^*(\overline{x})}$.
    Since $\calG \vDash G^*$, it holds that for every model $M$ and every $\overline{a}$, there is some goal $G : \chc{\bot}{B_G(\overline{y})}{\psi_G(\overline{y})}\in \calG$ with $M \vDash (B^*(\overline{a}) \land \phi^*(\overline{a}))\rightarrow \exists \overline{y}(B_G(\overline{y}) \land \psi_G(\overline{y}))$.
    Since $B_G$ contains only uninterpreted predicate symbols, there is $G$-$G^*$-feasible substitution $\theta$ s.t.\ $M \vDash (B^*(\overline{a}) \land \phi^*(\overline{a})) \rightarrow (B_G \land \psi_G)\theta(\overline{a})$.
    Since $M$ and $\overline{a}$ were arbitrary, it holds that $\vDash (B^* \land \phi^*) \rightarrow \bigvee_{G \in \calG} \bigvee_{\theta \in \Theta(G, G^*)} (B_G \land \psi_G)\theta$. In particular, $\vDash \phi^* \rightarrow \bigvee_{G \in \calG} \bigvee_{\theta \in \Theta(G, G^*)} \psi_G\theta$, which proves the claim.
\end{proof}

\section{Symbolic Execution and Incorrectness Logic}
\label{sec:incorrectness}

\emph{Incorrectness logic} is a program logic to reason about failing program executions; it has been proposed as a formalism for contracts documenting cases in which a function is supposed to terminate with an error, or to concisely represent counterexamples to correctness~\citep{DBLP:journals/pacmpl/OHearn20}. Incorrectness logic can be seen as the dual of Hoare logic.

\begin{center}
\begin{minipage}{0.4\linewidth}
\begin{equation}
    \{ p(\bar x) \} \,S\, \{ q(\bar x) \}
    \label{eq:hoare}
\end{equation}    
\end{minipage}
\qquad
\begin{minipage}{0.4\linewidth}
\begin{equation}
    [ p(\bar x) ] \,S\, [ q(\bar x) ]
    \label{eq:incorrect}
\end{equation}    
\end{minipage}
\end{center}
For a Hoare triple~\eqref{eq:hoare} to be valid, it has to be the case that the \emph{postcondition}~$q(\bar x)$ holds in every final state~$\bar x'$ of the statement~$S$ when started in a state~$\bar x$ satisfying the \emph{precondition}~$p(\bar x)$~\cite{Hoare1969}; in other words $q(\bar x)$ is implied by the strongest postcondition of $p(\bar x)$ and $S$. In contrast, an incorrectness logic triple~\eqref{eq:incorrect} is valid if for every state~$\bar x'$ satisfying the \emph{result}~$q(\bar x)$, there is some state~$\bar x$ satisfying the \emph{presumption}~$p(\bar x)$ from which the statement~$S$ terminates in state~$\bar x'$; the result~$q(\bar x)$ has to imply the strongest postcondition. If the semantics of the statement~$S$ is represented as a formula~$T_S(\bar x, \bar x')$ over two copies of the state variables, then the meaning of a Hoare triple is expressed by \eqref{eq:HoareAsFormula} and the meaning of an incorrectness triple by \eqref{eq:incorrectAsFormula}:
\begin{gather}
    \forall \bar x. \forall \bar x'.\;
    p(\bar x) \wedge T_S[\bar x, \bar x'] \to q(\bar x') \label{eq:HoareAsFormula}\\
    \forall \bar x'.\; q(\bar x') \to \exists \bar x.\; p(\bar x) \wedge T_S[\bar x, \bar x'] \label{eq:incorrectAsFormula}
\end{gather}

Since~\eqref{eq:HoareAsFormula} is Horn, it can naturally be encoded in the form of CHCs, making it possible to infer pre- and postconditions using CHC solvers. \emph{It is therefore tempting to ask whether a similar encoding is possible for incorrectness logic, so that the existing machinery of CHC can be put to use also for reasoning about incorrectness?}

The existential quantifier in \eqref{eq:incorrectAsFormula} indicates that such an encoding is not straightforward. Indeed, it is easy to see that in the special case of a statement~$S$ being \emph{injective}, i.e., $\forall \bar x, \bar x'_1, \bar x'_2.\; T_S(\bar x, \bar x'_1) \wedge T_S(\bar x, \bar x'_2) \to \bar x'_1 = \bar x'_2$, formula~\eqref{eq:incorrectAsFormula} can be expressed as a CHC. In the general case, to express \eqref{eq:incorrectAsFormula} an extended language of Horn clauses with existential quantifiers would be needed, which is beyond the scope of CHC and cannot be handled by existing CHC solvers.

We would like to advocate an alternative perspective on the nature of incorrectness logic. For this, observe that no abstraction takes place in the translation from a Hoare triple~\eqref{eq:hoare} to a Horn formula~\eqref{eq:HoareAsFormula}, the formula expresses \emph{exactly} the semantics of the program and the meaning of pre- and postconditions. It is then not clear why different translations to logic would be needed for correctness vs.\ incorrectness reasoning, instead the \emph{reasoning approach} itself can be replaced. It is well-known that the construction of Hoare logic proofs can be automated with the help of \emph{model construction} for sets of CHCs. We claim that the construction of incorrectness proofs can be supported through clausal proofs about CHCs, for instance with the help of the proof rules introduced in this article:
\begin{center}
    \begin{tabular}{ll}
        \emph{Logic} & \emph{Reasoning used for automation}
        \\\toprule
        Hoare logic & Model construction for sets of CHCs \\\hline
        Incorrectness logic & Proof construction for sets of CHCs
    \end{tabular}
\end{center}

Consider a forward proof rule for CHCs corresponding to \eqref{eq:hoare}, in which $\mathit{post}_S$ is some procedure approximating the post-image of $\varphi$ under the statement~$S$:
\begin{equation}
    \label{eq:schematicForward}
    \resolvecond{
    \chcfact{p(\bar x)}{\varphi[\bar x]}
    \qquad
    \chc{q(\bar x')}{p(\bar x)}{T_S[\bar x, \bar x']}
    }{
    \chcfact{q(\bar x')}{\psi[\bar x']}
    }{
    \psi = \mathit{post}_S(\varphi)
    }
\end{equation}
For the proof rule to be \emph{sound,} i.e., for the entailment
\begin{equation}
    \label{eq:schematicForwardSound}
    \{\forall \bar x.\, (\chcfact{p(\bar x)}{\varphi[\bar x]})
    ,
    \forall \bar x, \bar x'.\,
    (\chc{q(\bar x')}{p(\bar x)}{T_S[\bar x, \bar x']})\}
    ~\models~
    \forall \bar x'.\, (\chcfact{q(\bar x')}{\psi[\bar x']})
\end{equation}
to hold, it is necessary that the constraints $\varphi[\bar x], \psi[\bar x'], T_S[\bar x, \bar x']$ are in relationship:
\begin{lemma}
    \label[lemma]{lem:forwardSoundness}
    Proof rule~\eqref{eq:schematicForward} is sound if and only if $\psi[\bar x'] \Rightarrow \exists \bar x.\, \varphi[\bar x] \wedge T_S(\bar x, \bar x')$.
\end{lemma}
\begin{proof}
    ``$\Rightarrow$'' Assume that \eqref{eq:schematicForwardSound} holds, which means that it holds in particular for structures satisfying $p(\bar x) \leftrightarrow \varphi[\bar x]$ and $q(\bar x') \leftrightarrow \exists \bar x.\, \varphi[\bar x] \wedge T_S(\bar x, \bar x')$. The antecedent of \eqref{eq:schematicForwardSound} is valid in such structures, while the consequence is $\forall \bar x'.\, \psi[\bar x'] \to \exists \bar x.\, \varphi[\bar x] \wedge T_S(\bar x, \bar x')$.

    ``$\Rightarrow$'' Observe that the entailment holds when adding $\forall \bar x'.\, \psi[\bar x'] \to \exists \bar x.\, \varphi[\bar x] \wedge T_S(\bar x, \bar x')$ to the antecedent.
\end{proof}

The necessary and sufficient condition~$\psi[\bar x'] \Rightarrow \exists \bar x.\, \varphi[\bar x] \wedge T_S(\bar x, \bar x')$ for \eqref{eq:schematicForward} to be sound corresponds to the incorrectness formula~\eqref{eq:incorrectAsFormula}. This implies that forward symbolic execution proofs, such as the ones constructed in \Cref{sec:forward}, can directly be translated to proofs in incorrectness logic by using the derived constrained positive unit clauses (facts) as the result of incorrectness triples. Since backward proofs (\Cref{sec:backward}) can be translated to forward proofs, the same holds for backward symbolic execution.

This kind of reasoning is not limited to ground proofs, since symbolic execution naturally produces proofs that do not refer to concrete values, and hence gives rise to symbolic incorrectness proofs as well. CHC proofs can also be extended with the more specialized proof rules that have been introduced for incorrectness logic. As an example, consider the following ``Backwards Variant'' incorrectness rule:
\begin{equation}
    \resolvecond{
    [p(n) \wedge n \geq 0]\, \texttt{S}\, [p(n+1) \wedge n \geq 0]
    }{
    [p(0)]\, \texttt{while (*) S}\, [\exists n.\, p(n) \wedge n \geq 0]
    }{n : \mathbbm{Z} \text{~fresh}}
    \tag{Backwards Variant}
\end{equation}
The rule expresses that a loop \texttt{while (*) S} can be accelerated if it can be shown that $p(n)$ is a ``variant'' of the loop; $p(n)$ is required to decrease when executing the loop body~\texttt{S} backwards. We can formulate a corresponding unit propagation rule on CHCs, which enables the calculus to generalize a fact~$\chcfact{p(\bar x)}{\varphi[\bar x, 0]}$ to $\chcfact{p(\bar x)}{\varphi[\bar x, n] \wedge n \geq 0}$ if it can be shown, through a side condition, that $\varphi[\bar x, n]$ is a variant of $S$:
\begin{equation}
    \resolvecond{
    \begin{array}{c}
        \chcfact{p(\bar x)}{\varphi[\bar x, 0]} \\
        \chc{p(\bar x')}{p(\bar x)}{T_S[\bar x, \bar x']}
    \end{array}
    }{
    \chcfact{p(\bar x)}{\varphi[\bar x, n] \wedge n \geq 0}
    }{
    \begin{array}{l}
         n : \mathbbm{Z} \text{~fresh,} \\
    \varphi[\bar x, n+1] \wedge n \geq 0 \Rightarrow \mathit{post}_S(\varphi[\bar x, n])
    \end{array}
    }
\end{equation}
A more refined version of the proof rule could recover the condition~$\varphi[\bar x, n+1] \wedge n \geq 0\break \Rightarrow \mathit{post}_S(\varphi[\bar x, n])$ from the clausal proof itself. Whenever a proof contains a derivation of $\chcfact{p(\bar x)}{\varphi[\bar x, k+1]}$ from $\chcfact{p(\bar x)}{\varphi[\bar x, k]}$, for some $k \in \mathbbm{N}$, it can be checked whether the derivation can be generalized to a variable~$n$. In this case, thanks to \Cref{lem:forwardSoundness}, the derivation can be lifted to produce the more general consequence~$\chcfact{p(\bar x)}{\varphi[\bar x, n] \wedge\break n \geq 0}$. Experimenting with calculi of this kind is left as future work.

\section{Experiments}\label{sec:experiments}
We implemented forward and backward symbolic execution in \eldarica~\citep{DBLP:conf/fmcad/HojjatR18} and evaluate them against \eldarica's CEGAR engine on benchmarks from the CHC-COMP benchmark suite~\citep{DBLP:journals/corr/abs-2404-14923}.

\subsection{Setup}

We evaluate five configurations:
\begin{itemize}
  \item \eldabs: \eldarica's CEGAR engine with default settings.
  \item \eldabsoff: \eldarica's CEGAR engine with interpolation abstraction disabled.
  \item \eldsymexbfs: Forward execution (\texttt{FE}, \Cref{sec:forward}) via constrained positive unit hyper-resolution with BFS strategy.
  \item \eldsymexdfs: Forward execution with DFS strategy. This configuration only supports linear CHCs.
  \item \eldsymexsld: Backward execution (\texttt{BE}, \Cref{sec:backward}) via SLD resolution with BFS strategy. 
\end{itemize}
The symbolic execution configurations do not yet implement subsumption beyond checking whether an identical resolvent already exists.

We evaluate on nine categories from the CHC-COMP 2026 benchmark suite~\citep{DBLP:journals/corr/abs-2404-14923}.\footnote{\url{https://github.com/chc-comp/chc-comp-2026}} The categories are: ADT-LIA, ADT-LIA-Arrays, BV, BV-Lin, LIA, LIA-Arrays, LIA-Lin, LIA-Lin-Arrays, and LRA-Lin. Categories suffixed with ``Lin'' contain only linear CHCs.

Experiments were run on the Pelle cluster at UPPMAX, Uppsala University (80~nodes, AMD EPYC 9454P 48-core Zen4, 768\,GB RAM per node). Each benchmark was allocated 1~CPU core and 6\,GB of memory, with a wall-time limit of 180~seconds.

\subsection{Results and Discussion}
\begin{table}[tb]
  \centering
  \caption{Number of satisfiable\,/\,unsatisfiable instances solved per category. Best values for each category across configurations are shown in \textbf{bold}.
    }
  \label{tbl:results-overview}
  \footnotesize
  \begin{tabular}{lr r@{\,/\,}r r@{\,/\,}r r@{\,/\,}r r@{\,/\,}r r@{\,/\,}r}
    Category & \# & \multicolumn{2}{c}{\eldabs} & \multicolumn{2}{c}{\eldabsoff} & \multicolumn{2}{c}{\eldsymexbfs} & \multicolumn{2}{c}{\eldsymexdfs} & \multicolumn{2}{c}{\eldsymexsld} \\
    \cmidrule(lr){1-1} \cmidrule(lr){2-2} \cmidrule(lr){3-4} \cmidrule(lr){5-6} \cmidrule(lr){7-8} \cmidrule(lr){9-10} \cmidrule(lr){11-12}
    ADT-LIA & 1131 & \textbf{76} & 158 & 55 & \textbf{165} & 11 & 155 & \multicolumn{2}{c}{--} & 20 & 148 \\
    ADT-LIA-Arrays & 3607 & \textbf{2256} & \textbf{1123} & 2254 & 1122 & 1498 & 1070 & \multicolumn{2}{c}{--} & 2030 & 1090 \\
    BV & 96 & \textbf{33} & 42 & 31 & \textbf{43} & 8 & 26 & \multicolumn{2}{c}{--} & 2 & 33 \\
    LIA & 1287 & \textbf{633} & \textbf{378} & 571 & \textbf{378} & 91 & 365 & \multicolumn{2}{c}{--} & 166 & 293 \\
    LIA-Arrays & 1738 & \textbf{1029} & \textbf{684} & \textbf{1029} & \textbf{684} & 709 & 622 & \multicolumn{2}{c}{--} & 955 & 665 \\
    \midrule
    BV-Lin & 512 & \textbf{97} & \textbf{152} & 74 & \textbf{152} & 45 & 148 & 33 & 99 & 28 & \textbf{152} \\
    LIA-Lin & 1406 & \textbf{540} & \textbf{293} & 466 & \textbf{293} & 151 & 282 & 56 & 48 & 195 & 262 \\
    LIA-Lin-Arrays & 139 & \textbf{53} & \textbf{23} & 48 & \textbf{23} & 17 & 22 & 12 & 11 & 14 & 20 \\
    LRA-Lin & 277 & 27 & \textbf{13} & 30 & \textbf{13} & \textbf{34} & \textbf{13} & 2 & 4 & 28 & \textbf{13} \\
  \end{tabular}
\end{table}

\begin{table}[tbp]
  \caption{Comparison matrices for non-linear categories. Each entry shows sat/unsat instances solved by the row configuration but not by the column configuration.}
  \label{tbl:matrices-nonlinear}
  \centering
  \setlength{\tabcolsep}{3pt}
  \footnotesize
  \begin{tabular*}{\textwidth}{@{\extracolsep{\fill}}cc@{}}
      \makebox[0.48\textwidth][c]{ADT-LIA} & \makebox[0.48\textwidth][c]{ADT-LIA-Arrays} \\[.2em]
    \makebox[0.48\textwidth][c]{%
      \begin{tabular}{l cccc}
      & \eldabs & \eldabsoff & \eldsymexbfs & \eldsymexsld \\ \midrule
      \eldabs & --- & 22/1 & 65/13 & 56/14 \\
      \eldabsoff & 1/8 & --- & 44/14 & 35/17 \\
      \eldsymexbfs & 0/10 & 0/4 & --- & 0/10 \\
      \eldsymexsld & 0/4 & 0/0 & 9/3 & --- \\
      \end{tabular}
    }
    &
    \makebox[0.48\textwidth][c]{%
      \begin{tabular}{l cccc}
      & \eldabs & \eldabsoff & \eldsymexbfs & \eldsymexsld \\ \midrule
      \eldabs & --- & 2/1 & 769/56 & 226/34 \\
      \eldabsoff & 0/0 & --- & 767/56 & 224/34 \\
      \eldsymexbfs & 11/3 & 11/4 & --- & 20/8 \\
      \eldsymexsld & 0/1 & 0/2 & 552/28 & --- \\
      \end{tabular}
    }
    \\[4em]
    \makebox[0.48\textwidth][c]{BV} & \makebox[0.48\textwidth][c]{LIA}\\[.2em]
    \makebox[0.48\textwidth][c]{%
      \begin{tabular}{l cccc}
      & \eldabs & \eldabsoff & \eldsymexbfs & \eldsymexsld \\ \midrule
      \eldabs & --- & 3/0 & 25/17 & 31/10 \\
      \eldabsoff & 1/1 & --- & 23/17 & 29/10 \\
      \eldsymexbfs & 0/1 & 0/0 & --- & 7/2 \\
      \eldsymexsld & 0/1 & 0/0 & 1/9 & --- \\
      \end{tabular}
    }
    &
    \makebox[0.48\textwidth][c]{%
      \begin{tabular}{l cccc}
      & \eldabs & \eldabsoff & \eldsymexbfs & \eldsymexsld \\ \midrule
      \eldabs & --- & 68/1 & 542/25 & 469/86 \\
      \eldabsoff & 6/1 & --- & 480/25 & 408/86 \\
      \eldsymexbfs & 0/12 & 0/12 & --- & 41/76 \\
      \eldsymexsld & 2/1 & 3/1 & 116/4 & --- \\
      \end{tabular}
    }
    \\[4em]
    \multicolumn{2}{c}{\makebox[0.48\textwidth][c]{LIA-Arrays}}\\[.2em]
    \multicolumn{2}{c}{\makebox[0.48\textwidth][c]{%
      \begin{tabular}{l cccc}
      & \eldabs & \eldabsoff & \eldsymexbfs & \eldsymexsld \\ \midrule
      \eldabs & --- & 0/0 & 320/62 & 74/19 \\
      \eldabsoff & 0/0 & --- & 320/62 & 74/19 \\
      \eldsymexbfs & 0/0 & 0/0 & --- & 2/5 \\
      \eldsymexsld & 0/0 & 0/0 & 248/48 & --- \\
      \end{tabular}
    }}   
  \end{tabular*}
\end{table}

\begin{table}[tbp]
  \caption{Comparison matrices for linear categories. Entry format as in \Cref{tbl:matrices-nonlinear}.}
  \label{tbl:matrices-linear}
  \centering
  \setlength{\tabcolsep}{3pt}
  \footnotesize
  \begin{tabular*}{\textwidth}{@{\extracolsep{\fill}}cc@{}}
  \makebox[0.48\textwidth][c]{BV-Lin} & \makebox[0.48\textwidth][c]{LIA-Lin} \\[.2em]
    \makebox[0.48\textwidth][c]{%
      \begin{tabular}{l ccccc}
      & \eldabs & \eldabsoff & \eldsymexbfs & \eldsymexdfs & \eldsymexsld \\ \midrule
      \eldabs & --- & 24/0 & 67/7 & 78/56 & 70/3 \\
      \eldabsoff & 1/0 & --- & 47/7 & 54/56 & 46/3 \\
      \eldsymexbfs & 15/3 & 18/3 & --- & 12/53 & 33/3 \\
      \eldsymexdfs & 14/3 & 13/3 & 0/4 & --- & 22/5 \\
      \eldsymexsld & 1/3 & 0/3 & 16/7 & 17/58 & --- \\
      \end{tabular}
    }
    &
    \makebox[0.48\textwidth][c]{%
      \begin{tabular}{l ccccc}
      & \eldabs & \eldabsoff & \eldsymexbfs & \eldsymexdfs & \eldsymexsld \\ \midrule
      \eldabs & --- & 80/3 & 409/42 & 484/245 & 360/49 \\
      \eldabsoff & 6/3 & --- & 338/41 & 413/245 & 294/49 \\
      \eldsymexbfs & 20/31 & 23/30 & --- & 95/234 & 72/27 \\
      \eldsymexdfs & 0/0 & 3/0 & 0/0 & --- & 13/1 \\
      \eldsymexsld & 15/18 & 23/18 & 116/7 & 152/215 & --- \\
      \end{tabular}
    }
    \\[4em]
    \makebox[0.48\textwidth][c]{LIA-Lin-Arrays} & \makebox[0.48\textwidth][c]{LRA-Lin}\\[.2em]
    \makebox[0.48\textwidth][c]{%
      \begin{tabular}{l ccccc}
      & \eldabs & \eldabsoff & \eldsymexbfs & \eldsymexdfs & \eldsymexsld \\ \midrule
      \eldabs & --- & 5/0 & 37/1 & 41/12 & 39/3 \\
      \eldabsoff & 0/0 & --- & 33/1 & 36/12 & 34/3 \\
      \eldsymexbfs & 1/0 & 2/0 & --- & 5/12 & 3/2 \\
      \eldsymexdfs & 0/0 & 0/0 & 0/1 & --- & 0/2 \\
      \eldsymexsld & 0/0 & 0/0 & 0/0 & 2/11 & --- \\
      \end{tabular}
    }
    &
    \makebox[0.48\textwidth][c]{%
      \begin{tabular}{l ccccc}
      & \eldabs & \eldabsoff & \eldsymexbfs & \eldsymexdfs & \eldsymexsld \\ \midrule
      \eldabs & --- & 0/0 & 4/0 & 25/9 & 3/0 \\
      \eldabsoff & 3/0 & --- & 5/0 & 28/9 & 5/0 \\
      \eldsymexbfs & 11/0 & 9/0 & --- & 32/9 & 8/0 \\
      \eldsymexdfs & 0/0 & 0/0 & 0/0 & --- & 0/0 \\
      \eldsymexsld & 4/0 & 3/0 & 2/0 & 26/9 & --- \\
      \end{tabular}
    }
  \end{tabular*}
\end{table}

\Cref{tbl:results-overview} summarizes the number of satisfiable and unsatisfiable instances solved by each configuration. \Cref{tbl:matrices-nonlinear,tbl:matrices-linear} show pairwise comparison matrices: each entry shows the number of satisfiable/unsatisfiable instances solved by the row configuration but not the column configuration.

The results show that forward and backward symbolic execution are often complementary: each strategy solves instances the other cannot in many categories. For satisfiability, \eldsymexsld is often more effective than \eldsymexbfs. %
While we expected symbolic execution to solve some unsatisfiable instances that CEGAR cannot, in many categories some satisfiable instances could only be solved by the symbolic execution configurations. For instance, 11 satisfiable instances in ADT-LIA-Arrays and 20 in LIA-Lin could only be solved by \eldsymexbfs.

DFS (\eldsymexdfs) solves fewer instances overall; however, it occasionally solves unsatisfiable instances that other strategies cannot. This is because DFS can find refutations deep within the branches it explores, but it can also get stuck exploring an infinitely deep branch that does not lead to a refutation.

Overall, the results demonstrate that symbolic execution of CHCs, even without subsumption beyond syntactic identity checks and without search heuristics, is a viable approach for CHC solving, while being substantially simpler to implement than CEGAR or IC3-based algorithms.

\section{Related Work}\label{sec:related}
We discuss the most closely related work on CHC solving algorithms and on resolution in logic programming.

\subsection{CHC Solving and Software Model Checking}
Most CHC solvers rely on algorithms that synthesize invariants.
\spacer~\cite{DBLP:conf/cav/KomuravelliGC14} extends the IC3 framework~\cite{DBLP:conf/vmcai/Bradley11} to CHCs, maintaining over-approximate procedure summaries computed using interpolation, and under-approximate reachability facts computed using model-based projection. \eldarica~\citep{DBLP:conf/fmcad/HojjatR18} implements CEGAR with predicate abstraction~\citep{DBLP:conf/cav/ClarkeGJLV00}. \golem~\citep{DBLP:conf/cav/BlichaBS23} supports several engines, including transition power abstraction (TPA)~\citep{DBLP:conf/tacas/BlichaFHS22}. TPA maintains a sequence of relations where the $n$-th relation over-approximates up to $2^n$ transition steps, computed using interpolation, allowing safety checks at exponentially increasing depth.

Constrained resolution, as studied in this article, avoids invariant synthesis. In forward execution, derived clauses encode reachable states; in backward execution, they encode proof obligations. \eldarica implements both forward and backward symbolic execution of CHCs, currently without subsumption beyond syntactic identity after simplification. \golem implements forward symbolic execution for linear CHCs, following the description in~\cite{DBLP:conf/vmcai/AbdullaLR24}.
In~\cite{DBLP:conf/vmcai/AbdullaLR24}, Abdulla et al. build on top of the symbolic execution engine in \eldarica formalized in the present article, adding graph neural network-guided search heuristics.

Bj{\o}rner et al. discusses both backward and forward derivations for CHCs, and recast $k$-induction as a fold/unfold transformation on linear CHCs~\cite{BjornerGMR15}. Backward resolution with subsumption, as studied in this article, generalizes $k$-induction and applies to non-linear clauses without requiring a separate transformation step.

Acceleration-driven clause learning~\cite{DBLP:conf/sas/FrohnG23},
implemented in the \loat tool~\cite{DBLP:conf/cade/FrohnG22} is a calculus for linear CHCs that is based on resolution. Loop acceleration techniques compute quantifier-free formulas characterizing the $N$-fold closure of a recursive clause, and the resulting accelerated clauses are added to the clause set. Resolution with an accelerated clause simulates arbitrarily many applications of the original recursive clause in a single step.
Our approach shares the property with acceleration-driven clause learning of being refutationally complete and can prove satisfiability, but does not necessarily terminate.

In the context of software model checking, \cite{DBLP:conf/sigsoft/DietschHMNP17} compares two approaches for predicate extraction in CEGAR: one based on Craig interpolation and one based on symbolic execution using strongest postconditions. The symbolic execution-based approach can be competitive, consistent with our finding that constrained resolution can be effective without interpolation.

\subsection{Resolution and Logic Programming}
Constrained resolution~\citep{DBLP:journals/ai/Burckert94} extends first-order resolution~\citep{DBLP:journals/jacm/Robinson65} by replacing syntactic unification with equality constraints over a background theory, delegating satisfiability checks to a constraint solver. In theory resolution~\citep{DBLP:journals/jar/Stickel85}, specialized reasoning procedures are instead incorporated directly into the resolution rule, which resolves on sets of literals whose conjunction is determined $T$-unsatisfiable by a decision procedure for the theory $T$. A survey of CHC analysis and transformation techniques from the constraint logic programming perspective is given in~\cite{DBLP:journals/tplp/AngelisFGHPP22}.

Forward execution corresponds to bottom-up evaluation of logic programs, and backward execution to top-down evaluation. SLD resolution~\citep{DBLP:conf/ifip/Kowalski74,DBLP:journals/jacm/AptE82}, the computational mechanism of Prolog, is a restriction of backward resolution to definite clauses with a selection rule. In Prolog, the goal is typically to compute answer substitutions for a query; extracting invariants or counterexamples is not a concern. In our setting, we are also interested in models and refutation proofs. Models of the CHCs serve as correctness proofs for programs, and refutation proofs correspond to concrete failing executions that can be used as proofs of bugs and for test-case generation. CHCs also involve a background theory, which makes semantic entailment between clauses undecidable in general and requires the subsumption approximations described in \Cref{sec:subsumption}.

\section{Conclusion}\label{sec:conclusion}
We studied constrained resolution as a framework for symbolic execution of CHCs. 
Several directions remain for future work. Concolic execution~\cite{DBLP:conf/pldi/GodefroidKS05} is a natural strategy to consider, but adapting it to CHCs is not straightforward: unlike programs, CHCs do not have a fixed notion of input variables to concretize. One option is to strengthen constraints when the SMT solver cannot decide satisfiability within a time limit. Other search strategies such as iterative deepening and generational search~\cite{DBLP:conf/ndss/GodefroidLM08} could be explored. In many cases checking subsumption could be useful, but the cost of subsumption checks depends on the background theory and can be expensive in practice. Bidirectional strategies that combine forward and backward execution could also be effective, since the two directions are complementary in practice. We hope that the unified view of symbolic execution and resolution over CHCs developed in this article can serve as a foundation for exploring these and other directions.

\bmhead{Acknowledgements}
The work was supported by the Swedish Research Council through grants~2021-06327 and~2025-06185, and by the Ethereum Foundation under grant~FY25-2124. The experimental evaluation was enabled by resources in project UPPMAX~2026/1-126 provided by Uppsala University at UPPMAX.

\bibliography{refs}

\end{document}